\documentclass[11pt]{article}

\usepackage[utf8]{inputenc}
\usepackage{amsmath}
\usepackage{amsthm}
\usepackage{amsfonts}
\usepackage{fullpage}
\usepackage{amssymb}
\usepackage{authblk}
\usepackage{mathtools}
\usepackage{thm-restate}
\usepackage{mleftright}
\usepackage{bm}
\usepackage[linesnumbered,ruled,lined,noend]{algorithm2e}
\usepackage[colorlinks=true, linkcolor=blue, citecolor=blue,urlcolor=black]{hyperref}
\usepackage[capitalize,noabbrev]{cleveref}
\usepackage[table]{xcolor}
\usepackage{etoolbox}
\usepackage{xparse}
\usepackage{nicefrac}
\usepackage{booktabs}
\usepackage{caption}
\usepackage{enumitem}
\usepackage{float}
\usepackage{array}
\usepackage{pifont}
\usepackage{mdframed}

\newcommand{\tilx}{\tilde{\vec{x}}}
\newcommand{\tils}{\tilde{\vec{s}}}
\newcommand{\tild}{\tilde{\vec{d}}}

\newcounter{qst}
\crefname{qst}{Question}{Questions}

\input{figs/tikzmacros.tex}

\newtheorem{definition}{Definition}
\newtheorem{lemma}{Lemma}

\newtheorem{proposition}{Proposition}
\newtheorem{theorem}{Theorem}
\newtheorem{corollary}{Corollary}
\newtheorem{assumption}{Assumption}

\DeclareMathOperator{\reg}{Reg}
\DeclareMathOperator{\tildereg}{{\tilde R}eg}

\usepackage{natbib}

\title{Near-Optimal No-Regret Learning Dynamics for General\\ Convex Games}
\author[1]{Gabriele Farina\footnote{Equal contribution.}}
\author[2]{Ioannis Anagnostides$^*$}
\author[3]{Haipeng Luo}
\author[4]{Chung-Wei Lee}
\author[5]{Christian Kroer}
\author[6]{Tuomas Sandholm}

\affil[1,2,6]{Carnegie Mellon University}
\affil[5]{Columbia University}
\affil[3,4]{University of Southern California}
\affil[6]{Strategy Robot, Inc.}
\affil[6]{Optimized Markets, Inc.}
\affil[6]{Strategic Machine, Inc.}

\affil[ ]{\texttt {\{gfarina,ianagnos\}@cs.cmu.edu}, \texttt {\{haipengl,leechung\}@usc.edu}\\ \texttt{christian.kroer@columbia.edu} and \texttt{sandholm@cs.cmu.edu}}

\makeatletter
    \patchcmd\algocf@Vline{\vrule}{\vrule \kern-0.4pt}{}{}
    \patchcmd\algocf@Vsline{\vrule}{\vrule \kern-0.4pt}{}{}
\makeatother
\SetKwComment{Hline}{}{\vspace{-3mm}\textcolor{gray}{\hrule}\vspace{1mm}}
\definecolor{darkgrey}{gray}{0.3}
\definecolor{commentcolor}{gray}{0.5}
\SetKwComment{Comment}{\color{commentcolor}[$\triangleright$\ }{}
\SetCommentSty{}
\SetNlSty{}{\color{darkgrey}}{}
\SetKwProg{Fn}{function}{}{}
\SetKwProg{Subr}{subroutine}{}{}
\crefalias{AlgoLine}{line}%
\crefname{algocf}{Algorithm}{Algorithms}

\newcommand*{\N}{{\mathbb{N}}}
\newcommand*{\R}{{\mathbb{R}}}

\newcommand{\lip}{L}
\newcommand{\regdep}{\mathfrak{R}}
\newcommand{\gradbnd}{B}
\newcommand{\mul}{\mu}
\newcommand{\ind}{r}
\newcommand{\tilu}{\tilde{\vec{u}}}
\newcommand{\cR}{\mathcal{R}}
\newcommand{\cA}{\mathcal{A}}
\newcommand{\cQ}{\mathcal{Q}}
\newcommand{\cZ}{\mathcal{Z}}
\newcommand{\cS}{\mathcal{S}}

\DeclareMathOperator{\polylog}{polylog}
\DeclareMathOperator{\poly}{poly}

\newcommand{\cX}{\mathcal{X}}
\newcommand{\defeq}{\coloneqq}
\renewcommand{\^}[1]{^{(#1)}}
\newcommand{\bbR}{\mathbb{R}}
\newcommand{\algoshort}{\texttt{LRL-OFTRL}\xspace}
\DeclareMathOperator*{\argmax}{arg\,max}
\DeclareMathOperator*{\argmin}{arg\,min}
\renewcommand{\vec}[1]{\bm{#1}}
\newcommand{\vstack}[2]{\begin{pmatrix} #1 \\ #2 \end{pmatrix}}

\newcommand{\ls}[1]{\tilde{\vec \ut}^{(#1)}}
\newcommand{\pz}{\tilde{\vec x}}
\newcommand{\p}[1]{\tilde{\vec x}^{(#1)}}
\newcommand{\ppz}{\tilde{\vec z}}
\newcommand{\pp}[1]{\tilde{\vec z}^{(#1)}}

\newcommand{\Ls}[1]{\tilde{\Ut}^{(#1)}}
\newcommand{\ru}[1]{\psi(#1)}
\newcommand{\bg}[2]{D_\psi\left(#1\,\middle\|\,#2\right)}
\newcommand{\ir}[1]{\left\langle#1\right\rangle}
\newcommand{\cp}{\tilde{\vec x}}

\newcommand{\nms}[2]{\left\|#1\right\|_{#2}^2}

\newcommand{\pnms}[2]{\nms{#1}{\nabla^2\ru{#2}}}

\newcommand{\etamin}{\frac{1}{50}}

\newcommand{\range}[1]{[\![#1]\!]}
\newcommand{\ut}{\vec{u}}
\newcommand{\Ut}{\vec{U}}
\newcommand{\vx}{\vec{x}}
\newcommand{\vg}{\vec{g}}

\newcommand{\mat}[1]{\mathbf{#1}}

\newcommand{\OFTRL}{\texttt{OFTRL}\xspace}

\def\[#1\]{%
  \begin{align*}%
    #1%
  \end{align*}%
}
\NewDocumentCommand{\numberthis}{om}{%
\IfNoValueTF{#1}{%
    \refstepcounter{equation}\tag{\theequation}%
  }{%
    \tag{#1}%
  }%
  \label{#2}%
}

\newcommand{\pnt}[1]{\mleft\|#1\mright\|_t} % Primal norm at time t
\newcommand{\dnt}[1]{\mleft\|#1\mright\|_{*,t}} % Dual norm at time t

\begin{document}

\maketitle

\begin{abstract}
      A recent line of work has established uncoupled learning dynamics such that, when employed by all players in a game, each player's \emph{regret} after $T$ repetitions grows polylogarithmically in $T$, an exponential improvement over the traditional guarantees within the no-regret framework. However, so far these results have only been limited to certain classes of games with structured strategy spaces---such as normal-form and extensive-form games. The question as to whether $O(\polylog T)$ regret bounds can be obtained for general convex and compact strategy sets---which occur in many fundamental models in economics and multiagent systems---while retaining efficient strategy updates is an important question.
      In this paper, we answer this in the positive by establishing the first uncoupled learning algorithm with $O(\log T)$ per-player regret in general \emph{convex games}, that is, games with concave utility functions supported on arbitrary convex and compact strategy sets. Our learning dynamics are based on an instantiation of optimistic follow-the-regularized-leader over an appropriately \emph{lifted} space using a \emph{self-concordant regularizer} that is peculiarly not a barrier for the feasible region.
      Our learning dynamics are efficiently implementable given access to a proximal oracle for the convex strategy set, leading to $O(\log\log T)$ per-iteration complexity; we also give extensions when access to only a \emph{linear} optimization oracle is assumed. Finally, we adapt our dynamics to guarantee $O(\sqrt{T})$ regret in the adversarial regime.
      Even in those special cases where prior results apply, our algorithm improves over the state-of-the-art regret bounds either in terms of the dependence on the number of iterations or on the dimension of the strategy sets.

      %Yet, many fundamental models in economics and multiagent systems involve more complex and general strategy spaces, begging the question as to whether $O(\log T)$ is attainable in those cases. In this paper, we establish the first uncoupled learning dynamics with $O(\log T)$ regret for general \emph{convex} games. Our learning dynamics are based on \emph{optimistic follow the regularized leader}, but operating over an appropriately lifted space. Further, they are efficiently implementable as long as we assume a \emph{quadratic approximation oracle}, or even a \emph{linear optimization oracle} for the set. At the same time, we guarantee $O(\sqrt{T})$ regret in the adversarial regime as well.
\end{abstract}

\section{Introduction}
\label{section:introduction}

\emph{Regret minimization} is a celebrated framework that has been central in the development of online learning and the theory of multiagent systems. Indeed, fundamental connections have been forged between no-regret learning and game-theoretic solution concepts~\citep{Freund99:Adaptive, Hart00:Simple,Foster97:Calibrated,Roughgarden15:Intrinsic}. More broadly, regret is an intrinsic measure of performance in online learning and games. Furthermore, regret minimization algorithms have enjoyed a remarkable practical success, being a primary component in recent landmark results in AI~\citep{Bowling15:Heads,Moravvcik17:DeepStack, Brown17:Superhuman,Brown19:Superhuman}. These advances were guided by game-theoretic principles, made possible by training the AI agents using \emph{self-play} under regret-minimizing algorithms, an approach that has proven to be more scalable compared to linear programming techniques. Nevertheless, the traditional no-regret framework is overly pessimistic, insisting on modeling the environment in a fully adversarial way. While this well-understood worst-case view might be justifiable for applications such a security games, it could be far from optimal in more benign and \emph{predictable} environments, including the setting of training agents using self-play. This begs the question: \emph{What are the optimal performance guarantees we can obtain when learning agents are competing against each other in general games?}

This fundamental question was first formulated and addressed by~\citet{Daskalakis11:Near} within the context of \emph{zero-sum games}. Since then, there has been a considerable interest in extending their guarantee to more general settings~\citep{Rakhlin13:Optimization,Syrgkanis15:Fast,Foster16:Learning,Chen20:Hedging,Daskalakis21:Fast,Piliouras22:Optimal}. In particular, \citet{Daskalakis21:Near} recently established that when all players in a general \emph{normal-form game} employ an \emph{optimistic} variant of \emph{multiplicative weights update (MWU)}, the regret of each player grows \emph{nearly-optimally} as $O(\log^4 T)$ after $T$ repetitions of the game, leading to an \emph{exponential improvement} over the guarantees obtained using traditional techniques within the no-regret framework. However, while normal-form games are a common way to represent strategic interactions in theory, most settings of practical significance inevitably involve more complex strategy spaces. For those settings, any faithful approximation of the game using the normal form is typically inefficient, requiring an action space that is exponential in the natural parameters of the problem, thereby limiting the practical implications of those prior results. This motivates our central question:\vspace{.4cm}

\begin{tikzpicture}
    \node[text width=14cm,align=center,inner sep=0pt] at (0,-0.3) {\textit{Can we establish near-optimal, efficiently implementable, and strongly uncoupled no-regret learning dynamics in general convex games?}};
    \node at (7.7cm,-0.3) {\refstepcounter{qst}\label{qst}($\clubsuit$)};
    \node at (-7.7cm,-0.6) {};
\end{tikzpicture}
\vspace{.4cm}

\emph{Convex games} are a rich class of games wherein the strategy space of each player is an arbitrary convex and compact set, while the utility of each player is an arbitrary concave function (see \cref{section:prel} for a formal description). As such, convex games encompass normal-form and \emph{extensive-form games}, but go well-beyond
to many other fundamental settings in economic theory including routing games, resource allocation problems, and competition between firms. Our primary contribution in this paper is to substantially extend prior results to all such games, addressing \cref{qst}.

\subsection{Our Contributions}

In this paper we introduce a novel no-regret learning algorithm, which we coin \emph{lifted log-regularized optimistic follow the regularized leader (\algoshort)}. \algoshort settles \cref{qst} in the positive, as summarized in the following theorem.\footnote{For simplicity in the exposition we use the $O(\cdot)$ notation in our introduction to suppress time-independent parameters that depend (polynomially) on the game; precise statements are deferred to \cref{section:main}.}

\begin{theorem}[Detailed version in \cref{theorem:main-detailed}]
    \label{theorem:main}
    Consider any general convex game. When all players employ our strongly uncoupled learning dynamics ($\algoshort$), the regret of each player grows as $O(\log T)$. At the same time, if the player is facing adversarial utilities we guarantee $O(\sqrt{T})$ regret.
\end{theorem}

Importantly, our learning dynamics are efficiently implementable given access to a \emph{proximal oracle} for the set (\cref{eq:quad_approx}), requiring only $O(\log \log T)$ operations per iteration (\Cref{theorem:prox-Newton}); such an oracle is weaker than the---relatively standard in convex optimization---\emph{quadratic optimization oracle}. We also point out extensions under a weaker \emph{linear optimization oracle}, albeit with a worse per-iteration complexity (\Cref{theorem:FW}). Our no-regret learning dynamics imply the first efficiently implementable and near-optimal regret guarantees in general convex games, significantly extending the scope of prior $O(\polylog T)$-regret guarantees~\citep{Daskalakis21:Near,Farina22:Kernelized}; a comparison with prior approaches is included in~\cref{table:results}. We remark that \cref{theorem:main} establishes near-optimal regret both under \emph{self-play}, and in the adversarial regime---meaning that the other players act so as to minimize the player's utility; the latter feature of adversarial robustness has been a central desideratum in this line of work (\emph{e.g.}, see the discussion in~\citep{Kangarshahi18:Lets,Daskalakis11:Near}).

Our proposed learning dynamics lie within the general framework of \emph{optimistic} no-regret learning, pioneered by~\citet{Chiang12:Online} and \citet{Rakhlin13:Optimization}. We leverage the \OFTRL algorithm of~\citet{Syrgkanis15:Fast}, but with some important twists. First, as detailed in \cref{algo:ours}, the \OFTRL optimization step is performed over a ``lifted'' space. While prior work in online learning has employed similar in spirit approaches~\citep{Lee20:Bias,Luo22:Adaptive}, our lifting is quite different, ensuring that the regret incurred by \OFTRL is \emph{nonnegative} (\cref{theorem:nonnegative}). Further, we employ a \emph{logarithmic self-concordant regularizer}; interestingly, and perhaps surprisingly, this is not a \emph{barrier} for the underlying feasible set. This deviates substantially from the typical use of self-concordant regularization (especially within the bandit setting~\citep{Abernethy08:Competing,Wei18:More,Bubeck19:Improved}). A pictorial overview of our construction is given in the caption of \Cref{algo:ours}.

The use of the logarithmic regularizer serves two main purposes. First, we show that it guarantees \emph{multiplicative stability} of the strategies, a refined notion of stability that is also leveraged in the work of~\citet{Daskalakis21:Near}. Nonetheless, we are the first to leverage such properties in general domains, going well beyond the guarantees of (Optimistic) MWU on the simplex~\citep{Daskalakis21:Near}. Further, the \emph{local norm} induced by the logarithmic regularizer enables us to cast regret bounds from the lifted space to the original space, while preserving the \emph{RVU property}~\citep[Definition 3]{Syrgkanis15:Fast}. In turn, this implies near-optimal regret by establishing that the \emph{second-order path lengths} up to time $T$ are bounded by $O(\log T)$ (\cref{theorem:traj}), building on a recent technique of~\citet{Anagnostides22:Uncoupled} which crucially leverages the nonnegativity of swap regret.\footnote{To see why nonnegativity is crucial, note that the RVU bound implies optimal \textit{sum} of players' regrets~\citep{Syrgkanis15:Fast}. Thus, nonnegativity would imply the same bound for each player's regret.}

\begin{table}[t]%\centering
     \newcommand{\ldarrow}{\raisebox{-.7mm}{\tikz \draw[->] (0,0) -- (.25,0) -- +(0, -.2);}}%
    \scalebox{.876}{
    \begin{tabular}{m{4cm}m{3.8cm}m{3.2cm}m{5.9cm}}
        \bf Method                   & \bf Applies to                              & \bf Regret bound                              & \bf Cost per iteration         \\
        \toprule
        OFTRL / OMD\newline\citep{Syrgkanis15:Fast}               & General convex set                         & $O(\sqrt{n}\, \regdep T^{1/4})$ & Regularizer- \& oracle- dependent          \\
        \midrule
        OMWU\newline\citep{Daskalakis21:Near}\!             & Simplex $\Delta^d$                          & $O(n \log d \log^4 T)$                    & $O(d)$                         \\
        \midrule
        Clairvoyant MWU\newline\citep{Piliouras22:Optimal}             & Simplex $\Delta^d$  & $O(n \log d)$\newline Subsequence only~$\ddagger$ & $O(d)$                                                                                                  \\
        \midrule
        Kernelized OMWU\newline\citep{Farina22:Kernelized}        & Polytope $\Omega = \textrm{co} \mathcal{V}$\newline with $\mathcal{V}\subseteq \{0,1\}^d$ & $O(n \log |\mathcal{V}| \log^4 T)$       & $d \times \text{cost of kernel}$\\
        \midrule
        \rowcolor{gray!20}\algoshort\newline\textbf{[This paper]} & General convex set\newline $\cX\subseteq \bbR^d$                          & $O(n d \|\cX\|_1^2 \log T)$                  & Oracle-dependent:\newline
            $\bullet$~~$O(\log\log T)$ proximal oracle calls\newline
            $\bullet$~~$O(\poly T)$ linear opt. oracle calls\\
        \bottomrule
    \end{tabular}}
    \caption{Comparison of prior results on minimizing external regret in games. For simplicity, we have suppressed dependencies on the smoothness and the range of the utilities. 
    We use $n$ to denote the number of players; $T$ to denote the number of repetitions; $\regdep$ to indicate a parameter that depends on the regularizer;
    $\textrm{co} \mathcal{V}$ to denote the convex hull of $\mathcal{V}$; and $\|\cX\|_1$ to denote a bound on the maximum $\ell_1$ norm of any strategy. $\ddagger$ Unlike all other algorithms, the full sequence of iterates produced by Clairvoyant MWU (CMWU) is not known to achieve sublinear regret. Rather, after running CMWU for $T$ iterations, only a smaller subsequence of length $\Theta(T/\log T)$ iterates is known to attain the regret stated in the table. So, we remark that in order to achieve a comparable approximation of a coarse correlated equilibrium, CMWU needs to be run for $\Theta(T \log T)$ iterations.
    % Unlike the other algorithms, Clairvoyant MWU is not uncoupled (see \Cref{footnote:Piliouras}); as such, it is not clear how to compare its per-iteration cost with the other methods since it does not perform parallel strategy updates. $\dagger$ \citet[Remark 5]{Piliouras21:Optimal} suggest that some of their techniques could apply beyond simplex domains, though it is unclear whether the regret bound and efficient implementation of the dynamics would transfer.
    }
    \label{table:results}
\end{table}

\subsection{Further Related Work}

The rich line of work pursuing improved regret guarantees in games was pioneered by~\citet{Daskalakis11:Near}. Specifically, they developed \emph{strongly uncoupled} learning dynamics so that the players' regrets grow as $O(\log T)$, an exponential improvement over the guarantee one could hope for in adversarial environments~\citep{Shalev-Shwartz12:Online,Cesa-Bianchi06:Prediction}. Their result was significantly simplified by~\citet{Rakhlin13:Optimization}---again in zero-sum games---who introduced a simple variant of \emph{mirror descent} with a \emph{recency bias}---a.k.a. \emph{optimistic} mirror descent (OMD). It is worth noting that, beyond the benefits of optimism from an optimization standpoint~\citep{Polyak87:Introduction}, recency bias has been experimentally documented in natural learning environments in economics~\citep{Fudenberg14:Recency}.

Subsequently, \citet{Syrgkanis15:Fast} crystallized the \emph{RVU property}, an adversarial regret bound applicable for a broad class of optimistic no-regret learning algorithms. Using that property, they showed that the individual regret of each player grows as $O(T^{1/4})$ in general games, thereby converging to the set of \emph{coarse correlated equilibria} with a rate of $O(T^{-3/4})$. A near-optimal bound of $O(\polylog(T))$ in normal-form games was finally established by~\citet{Daskalakis21:Near}, while~\citet{Farina22:Kernelized} generalized that result in a class of polyhedral games that includes extensive-form games. Some extensions of the previous results have also been established for the stronger notion of \emph{no-swap-regret} learning dynamics in normal-form games~\citep{Chen20:Hedging,Anagnostides21:Near,Anagnostides22:Uncoupled}.
In particular, our work builds on a very recent technique of~\citet{Anagnostides22:Uncoupled}, which established $O(\log T)$ swap regret in normal-form games using as a regularizer a self-concordant \emph{barrier} function. 
%\todo{Gabri: I think we could explain how that paper is different. Maybe "builds on" is a bit ambiguous too, it might sound like an incremental change. Maybe we can say that the idea there is the nonnegativity of the regret, but point out that external regret can be negative, and so things are unclear. We solve that by working in a lifted space where regret is guaranteed nonnegative} 
On the other hand, establishing even sublinear $o(T)$ swap regret in extensive-form games is a notorious open question. Finally, an interesting new approach for obtaining near-optimal external regret in normal-form games was recently proposed in concurrent work by~\citet{Piliouras22:Optimal}.\footnote{An earlier version of the paper~\citep{Piliouras21:Optimal} proposed a preliminary and \emph{not uncoupled} version of the Clairvoyant MWU algorithm whose iterates were guaranteed to be no-regret and require $O(d\log T)$ per-iteration complexity. The 2022 revision of that paper provides an \emph{uncoupled} version with time-independent $O(d)$ per-iteration complexity, albeit at the cost of losing the no-regret guarantee on the entire sequence of iterates. See also footnote $\ddagger$ in \cref{table:results}.} %Nevertheless, their proposed dynamics are \emph{not} uncoupled; hence, their regret bounds are not comparable with the aforementioned results.%\footnote{Concurrently with this work, \citet{Piliouras22:Optimal} developed an uncoupled variant of Clairvoyant MWU with similar regret guarantees, albeit only for a subsequence of iterates.\label{footnote:Piliouras}}

Games with continuous strategy spaces have received a lot of attention in the literature; \emph{e.g.}, see \citep{Roughgarden15:Local,Even-Dar09:On,Harks11:Demand,Hsieh21:Adaptive,Mertikopoulos19:Learning,Stein11:Correlated,Stoltz07:Learning}, and references therein. Such games encompass a wide variety of applications in economics and multiagent systems; we give several examples in \Cref{section:prel}. Indeed, in many applications of interest a faithful approximation of the game requires an extremely large or even infinite action space; such settings could be abstracted as \emph{Littlestone games} in the sense of the recent work of~\citet{Daskalakis21:Fast}.

\section{No-Regret Learning and Convex Games}
\label{section:prel}

In this section we review the general setting of \emph{convex games}\footnote{Sometimes these are referred to as \emph{concave games}~\citep{Rosen65:Existence} or \emph{continuous games}~\citep{Hsieh21:Adaptive}.} which encompasses a number of important applications, as explained in \cref{section:applications}. We then formally define the framework of uncoupled and online no-regret learning in games in \cref{section:regret}.

\paragraph{Notation} We let $\N = \{1, 2, \dots, \}$ be the set of natural numbers. For a vector $\vec{x} \in \R^d$ we denote by $\vec{x}[r]$ its $r$-th coordinate, for some index $r \in \range{d} \defeq \{1, 2, \dots, d\}$. We will typically represent the players using subscripts; superscripts are reserved for the time index, denoted by the variable $t$.

\subsection{Convex Games}

Let $\range{n} \defeq \{1, 2, \dots, n\}$ be a set of players, with $n \in \N$. In a \emph{convex game}, every player $i \in \range{n}$ has a nonempty convex and compact set of strategies $\cX_i \subseteq \R^{d_i}$. For a \emph{joint strategy profile} $\vec{x} = (\vec{x}_1, \dots, \vec{x}_n) \in \bigtimes_{j=1}^n \cX_j$, the reward of player $i$ is given by a continuously differentiable utility function $u_i : \bigtimes_{j=1}^n \cX_j \to \R$ subject to the following standard assumption.
\begin{assumption}[Convex games]
      \label{assumption:smooth}
      The utility function $u_i(\vec x_1, \dots, \vec x_n)$ of any player $i \in \range{n}$ satisfies the following properties:
      \begin{enumerate}[nosep, leftmargin=5mm]
          \item (Concavity) $u_i(\vec{x}_i, \vec{x}_{-i})$ is \emph{concave} in $\vec{x}_i$ for $\vec{x}_{-i} = (\vec{x}_1, \dots, \vec{x}_{i-1}, \vec{x}_{i+1}, \dots, \vec{x}_n) \in \bigtimes_{j \neq i} \cX_j$;
          \item (Bounded gradients) for any $(\vec x_1, \dots, \vec x_n) \in \bigtimes_{j=1}^n \cX_j$, $\|\nabla_{\vec x_i} u_i(\vec x_1, \dots, \vec x_n) \|_\infty \le \gradbnd$, for some parameter $\gradbnd > 0$; and
          \item ($L$-smoothness) there exists $\lip > 0$ so that for any two joint strategy profiles $\vec{x}, \vec{x}' \in \bigtimes_{j=1}^n \cX_j$,
          \[
                \|\nabla_{\vec{x}_i} u_i(\vec{x}) - \nabla_{\vec{x}_i} u_i(\vec{x}') \|_\infty \leq L \sum_{j \in \range n} \| \vec{x}_j - \vec{x}_j' \|_1.
          \]
      \end{enumerate}
\end{assumption}
%We will also let $L \defeq \max_{i \in \range{n}} L_i$. %We point out that the norms in \eqref{eq:smooth} are selected so that $L \leq 1$ in normal-form games, but the definition could be translated to any equivalent pair of norms.

\subsection{Applications and Examples of Convex Games}
\label{section:applications}

Here we discuss several different classes of games which can all be analyzed under the common framework of convex games. For simplicity, we describe Cournot competion in the one-dimensional setting, but it can be readily generalized in more general domains. For more examples, we refer to~\citep{Even-Dar09:On,Hsieh21:Adaptive}, and references therein.

\begin{description}[itemsep=1mm,leftmargin=6mm]
\item[Normal-Form Games] In \emph{normal-form games (NFGs)} every player $i \in \range{n}$ has a finite and nonempty set of strategies $\cA_i$. Player $i$'s strategy set contains all probability distributions supported on $\cA_i$; that is, $\cX_i = \Delta(\cA_i)$. The utility of player $i$ can be expressed as the \emph{multilinear} function $u_i(\vec{x}) \defeq \mathbb{E}_{\vec{a} \sim \vec{x}} [\mathcal{U}_i(\vec{a})] %\prod_{j \in \range{n}} \vec{x}_j[a_j]
$, for some arbitrary function $\mathcal{U}_i: \bigtimes_{j=1}^n \cA_j \to \R$.

\item[Extensive-Form Games] \emph{Extensive-form games (EFGs)} generalize NFGs by capturing both sequential and simultaneous moves, stochasticity from the environment, as well as \emph{imperfect information}. EFGs are abstracted on a directed tree. Once the game reaches a terminal (or leaf) node $z \in \cZ$, each player $i \in \range{n}$ receives a utility $\mathcal{U}_i(z)$, for some $\mathcal{U}_i: \cZ \to \R$. The strategy space of each player $i \in \range{n}$ can be compactly represented using the \emph{sequence-form polytope} $\cQ_i$~\citep{Romanovskii62:Reduction,Koller96:Efficient}. If $p_c(z)$ is the probability of reaching terminal node $z \in \cZ$ over ``chance moves'', the utility of player $i$ can be expressed as $u_i(\vec{q}) \defeq \sum_{z \in \cZ} p_c(z) \mathcal{U}_i(z) \prod_{j \in \range{n}} \vec{q}_j[\sigma_{j, z}]$, where $\vec{q} = (\vec{q}_1, \dots, \vec{q}_n) \in \bigtimes_{j=1}^n \cQ_j$ is the joint strategy profile, and $\vec{q}_{j}[\sigma_{j, z}]$ is the probability mass assigned to the last \emph{sequence} $\sigma_{j, z}$ encountered by player $j$ before reaching $z$. The smoothness and the concavity of the utilities follow directly from multilinearity; for a more detailed account on EFGs we refer the interested reader to the excellent book of~\citet{Shoham08:Multiagent}.

\item[Splittable Routing Games] In these games~\citep{Roughgarden15:Local} every player has to route a flow $f_i$ from a source to a destination in an undirected graph $G = (V,E)$. Every edge $e \in E$ is associated with a latency function $\ell_e(f_e)$ mapping the amount of flow passing through the edge to some latency. The set of strategies of player $i$ corresponds to the possible ways of ``splitting'' the flow $f_i$ into paths from the source to the destination. Under suitable restrictions on the latency functions, those games satisfy \cref{assumption:smooth} (see~\citep{Syrgkanis15:Fast}).

\item[Cournot Competition] This game is played among $n$ firms (players). Every firm $i$ decides the \emph{quantity} $s_i \in \cS_i \subseteq \bbR_{\geq 0}$ of a common good to produce, where $\cS_i$ is an interval. Further, a cost function $c_i : \cS_i \to \R$ assigns a \emph{production cost} to a given quantity, while $p : \bigtimes \cS_i \to \R_{\geq 0}$ is the price of the good determined by the the joint choice of quantity $\vec{s} = (s_1, \dots, s_n)$ across the firms. Then, the utility of firm $i$ is defined as $u_i(\vec{s}) \defeq s_i p(\vec{s}) - c_i(s_i)$. In \emph{linear Cournot competition}, $p(\vec{s}) \defeq a - b \left( \sum_{i=1}^n s_i \right)$, for some $a, b > 0$, while the cost functions $c_i$ are assumed to be smooth and convex~\citep{Even-Dar09:On}.

%\item[Resource Allocation Games] In a \emph{resource allocation game} every player submits a bid $b_i \in [a_i, c_i] \subseteq \R$, where $0 < a_i \le c_i$. The vector of joint bids for all players is $\vec b \defeq (b_1,\dots,b_n)$. The allocation function $m : \bigtimes_{j=1}^n [a_j, c_j] \to \Delta^n$ maps the joint vector of bids to an allocation of resources $(m_1(\vec{b}), \dots, m_n(\vec b)) \in \Delta^n$. A \emph{value function} $v_i : [0, 1] \to \R$ assigns a utility to player $i$ given their allocation $m_i(\vec b)$, while $u_i : \bigtimes_{j=1}^n [a_j, c_j] \to \R$ represents $i$'s utility, defined as $u_i(\vec{b}) \defeq v_i(m_i(\vec{b})) - b_i$. In a \emph{linear allocation game}, $v_i(d_i) = \alpha_i d_i$, while $m_i(\vec{b}) \defeq b_i^\gamma / (\sum_{j = 1}^n b_j^\gamma)$, for some fixed parameters $\alpha_i \in \bbR$ and $0 < \gamma \le 1$~\citep{Even-Dar09:On,Johari04:Efficiency}.
\end{description}

\subsection{Online Linear Optimization and No-Regret Learning}
\label{section:regret}

In the \emph{online learning framework} a learning agent has to select a strategy $\vec{x}\^t \in \cX \subseteq \R^d$ at every time $t \in \N$. Then, in the \emph{full information} model, the learner receives as feedback from the environment a \emph{linear} utility function $\vec{x} \mapsto \langle \vec{x}, \vec{u}\^t \rangle$, for some vector $\vec{u}\^t \in \R^d$. The canonical measure of performance is the notion of \emph{regret}, defined for a time horizon $T \in \N$ as follows.
\begin{equation}
      \label{eq:linreg}
      \reg^T \defeq \max_{\vec{x}^* \in \cX} \left\{ \sum_{t=1}^T \langle \vec{x}^*, \vec{u}\^t \rangle \right\} - \sum_{t=1}^T \langle \vec{x}\^t, \vec{u}\^t \rangle.
\end{equation}
That is, the performance of the agent is compared to the optimal \emph{fixed} strategy in hindsight. It is important to note that \emph{regret can be negative}. In the context of convex games, it is assumed that every player $i \in \range{n}$ receives at time $t$ the ``linearized'' utility function $\vec{x}_i \mapsto \langle \vec{x}_i, \vec{u}_i\^t \rangle$, where $\vec{u}_i^{(t)} \defeq \nabla_{\vec{x}_i} u_i(\vec{x}\^t)$. By concavity (\Cref{assumption:smooth}),
\begin{equation*}
      \max_{\vec{x}_i^* \in \cX_i} \sum_{t=1}^T \left( u_i(\vec{x}_i^*, \vec{x}\^t_{-i}) - u_i(\vec{x}\^t) \right) \leq \max_{\vec{x}_i^* \in \cX_i} \sum_{t=1}^T \langle \vec{x}_i^{*} - \vec{x}_i\^t, \nabla_{\vec{x}_i} u_i(\vec{x}\^t) \rangle.
\end{equation*}
As a result, a regret bound on the linearized regret---in the sense of \eqref{eq:linreg}---automatically translates to a regret bound in the convex game. 

\paragraph{Strongly Uncoupled Learning Dynamics} In this setting, all learning dynamics are \emph{uncoupled} in the sense of ~\citet{Hart03:Uncoupled}: every player is oblivious to the other players' utilities. In fact, players need not have any prior knowledge about the game, even about their own utilities; this captures the condition of \emph{strong uncoupledness} of~\citet{Daskalakis11:Near}, along with a suitable bound on the memory of each player.

%For simplicity, and without any loss of generality, for the rest of this paper we will assume that $\|\vec{u}\^t_i \|_\infty\|\vec{x}_i \|_1 \leq 1$, for any time $t \in \N$, player $i \in \range{n}$, and $\vec{x}_i \in \cX_i$.
%$\vec{u}_i\^t \in \R^d_{\geq 0}$; $\langle \vec{x}_i, \vec{u}_i\^t \rangle \leq 1$ for any $\vec{x}_i \in \cX_i$; and $\|\vec{u}\^t_i \|_\infty \leq 1$, for any time $t \in \N$ and player $i \in \range{n}$.

%\section{Algorithm: Log-Regularized Lifted Optimistic FTRL}
\section{Near-Optimal No-Regret Learning in Convex Games}
\label{section:main}

In this section we describe our algorithm, \emph{Log-Regularized Lifted Optimistic FTRL} (henceforth $\algoshort$). The central result of this section, \cref{theorem:main-detailed}, asserts that when all players learn using \algoshort, their regret only grows logarithmically with respect to the number of repetitions of the game. Detailed proofs for this section are available in \cref{app:proofs}. 

\subsection{Setup}\label{sec:setup}
In the sequel, we will define and analyze the regret cumulated by \algoshort from the perspective of a generic player, omitting player subscripts.

We denote the set of strategies of the player by $\cX \subseteq \bbR^d$. %; then, assuming that all players follow our dynamics, we will establish \cref{theorem:main-detailed}.
Without loss of generality, we will assume that $\cX \subseteq [0,+\infty)^d$; otherwise, it suffices to first shift the set. Furthermore, we assume without loss of generality that there is no index $\ind \in \range d$ such that $\vec x[\ind] = 0$ for all $\vec x \in \cX$---if not, dropping the identically-zero dimension would not alter regret.
We define the \emph{lifting of set $\cX$} as the following set: %$\tilde \cX \subseteq \bbR^{d+1}$ defined as
\begin{equation}
    \label{eq:lifting}
    \bbR^{d+1}\supseteq \tilde\cX \defeq \mleft\{ (\lambda, \vec y) : \lambda \in [0,1], \vec y \in \lambda\cX \mright\}.
\end{equation}
Further, we define the $\ell_1$-norm $\| \cX \|_1$ of $\cX$ as the maximum $\ell_1$-norm of any vector $\vec x \in \cX$, that is, $\|\cX\|_1 \defeq \max_{\vec{x} \in \cX} \| \vec{x} \|_1$; for example, $\|\Delta^d\|_1 = 1$.

The \emph{logarithmic regularizer} for $\bbR^{d+1}$ is the function
\[\cR(\lambda, \vec{y}) \defeq - \log \lambda - \sum_{\ind=1}^d \log \vec{y}[\ind], \qquad \forall (\lambda,\vec y) \in \bbR^{d+1}_{> 0}.
\]
Given any vector $(\lambda, \vec y) \in \tilde\cX \cap \bbR_{>0}^{d+1}$, we denote with $\|\cdot\|_{(\lambda, \vec y)}$ and $\|\cdot\|_{*,(\lambda, \vec y)}$ the \emph{local norms centered at $(\lambda,\vec y)$} induced by $\cR(\lambda,\vec{y})$, defined as
\[
    \mleft\|\vstack{a}{\vec z}\mright\|_{(\lambda, \vec y)} \defeq \sqrt{\mleft(\frac{a}{\lambda}\mright)^2 + \sum_{\ind=1}^d \mleft(\frac{\vec z[\ind]}{\vec y[\ind]}\mright)^2}, \qquad
    \mleft\|\vstack{a}{\vec z}\mright\|_{*,(\lambda, \vec y)} \defeq \sqrt{(a \lambda)^2 + \sum_{\ind=1}^d \mleft(\vec z[\ind]\vec y[\ind]\mright)^2} %\qquad\qquad\mleft(\vstack{a}{\vec z} \in \bbR\times\bbR^d\mright).
\]
for any $(a,\vec z) \in \bbR^{d+1}$. These are the norms induced by the Hessian matrix of $\cR$ at $(\lambda, \vec y)$ and its inverse. It is a well-known fact that $\|\cdot\|_{*,(\lambda, \vec y)}$ is the dual norm of $\|\cdot\|_{(\lambda, \vec y)}$, and \emph{vice versa}.

\subsection{Overview of Our Algorithm} Our algorithm (Algorithm~\ref{algo:ours}) leverages \emph{optimistic follow the regularized leader} (\OFTRL), a simple variant of \texttt{FTRL} introduced by~\citet{Syrgkanis15:Fast}, but with some important twists. First, the optimization is performed over the lifting $\tilde{\cX}$ of the set $\cX$. More precisely, at every iteration the observed utility $\vec{u}\^t \in \R^d$ will be transformed to $\tilu\^t \in \R^{d+1}$ according to \cref{line:lift}; this ensures that $\tilu\^t$ is orthogonal to the vector $(1, \vec{x}\^t)$. Then, this utility vector $\tilu\^t$ is given as input to a regret minimizer operating over $\tilde{\cX}$, employing \OFTRL under the logarithmic regularizer $\cR(\lambda, \vec{y})$; this step is described in \cref{line:oftrl}. We discuss how such an optimization problem can be solved efficiently in \cref{sec:implementation}. Below we point out that \Cref{line:oftrl} is indeed well-defined. % The proof follows using standard arguments. %\emph{e.g.}, see \citep[Lemma 3]{Dinh15:Composite}.

\begin{restatable}{proposition}{propuniqoftrl}
    \label{proposition:unique}
    For any $\eta \ge 0$ and at all times $t \in \N$, the \OFTRL optimization problem on \cref{line:oftrl} of \cref{algo:ours} admits a unique optimal solution $(\lambda^{(t)}, \vec{y}^{(t)}) \in \tilde \cX \cap \bbR^{d+1}_{> 0}$.
\end{restatable}

Finally, given the iterate $(\lambda\^t, \vec{y}\^t)$ output by the \OFTRL step at time $t$, our regret minimizer over $\cX$ selects the next strategy $\vec{x}\^t \defeq {\vec{y}\^t}/{\lambda\^t}$ (\cref{line:norm}); this is indeed a valid strategy in $\cX$ by definition of $\tilde{\cX}$ in \eqref{eq:lifting}, as well as the fact that $\lambda\^t > 0$ as asserted in \cref{proposition:unique}.

\begin{algorithm}[htp]
    \caption{Log-Regularized Lifted Optimistic FTRL (\algoshort)\newline
        \makebox[\columnwidth][c]{\protect\scalebox{.93}{
    \begin{tikzpicture}[yscale=1.1,xscale=1.3]
        \fill[black!10] (0.6,0.3) rectangle (9.7,1.7) node[fitting node] (outer) {};
        \draw[fill=white,semithick] (1,0.5) rectangle (2.7,1.5) node[fitting node] (lifting) {};
        \node[text width=1cm,align=center] at (lifting.center) {Lifting};

        \draw[fill=white,semithick] (3.5,0.5) rectangle (5.5,1.5) node[fitting node] (oftrl) {};
        \node[text width=2.5cm,align=center] at (oftrl.center) {\OFTRL with log regularizer};

        \draw[fill=white,semithick] (7,0.5) rectangle (9,1.5) node[fitting node] (norm) {};
        \node[text width=2cm,align=center] at (norm.center) {Normalization};

        \draw[semithick,->] (-.0,1) node[above,xshift=3mm,yshift=1mm,inner ysep=0,fill=white]{$\ut\^t$} -- (lifting);
        \draw[semithick,->] (lifting) -- (oftrl) node[above,pos=.6] {$\tilde{\ut}\^t$};
        \draw[semithick,->] (oftrl) -- (norm) node[above,pos=.5] {\small$(\lambda\^t, \vec y\^t)$};
        \path (norm) -- +(2.5,0) node[above=.1mm,inner ysep=1.2mm,pos=.5] {$\vec x\^t \defeq \frac{\vec y\^t}{\lambda\^t}$};
        \draw[semithick,->] (norm) -- +(2.5,0);

        % \node[fill=white,inner ysep=0pt] at (outer.north) {\small\color{black}\cref{algo:ours} (\algoshort)};
    \end{tikzpicture}
}}\vspace{2mm}
    }\label{algo:ours}
    \DontPrintSemicolon
    \KwData{Learning rate $\eta$}\vspace{2mm}
    Set $\tilde{\Ut}\^1, \vec{u}^{(0)} \gets \vec{0} \in \bbR^{d+1}$\;
    \For{$t=1,2,\dots, T$}{
    Set  $\displaystyle \vstack{\lambda\^t}{\vec y\^t} \gets \argmax_{(\lambda, \vec y)\in\tilde\cX}\mleft\{ {\eta} \mleft\langle \tilde{\Ut}\^t + \tilde{\vec\ut}\^{t-1}, \vstack{\lambda}{\vec y} \mright\rangle + \log \lambda + \sum_{\ind=1}^d \log \vec{y}[\ind]\mright\}$\Comment*{\color{commentcolor}\texttt{OFTRL}]\!\!\!\!\!}\medskip \label{line:oftrl}
    Play strategy $\displaystyle\vec x\^t \defeq \frac{\vec y\^t}{\lambda\^t} \in \cX$\Comment*{\color{commentcolor}Normalization]\!\!\!\!\!} \label{line:norm}
    Observe $\vec\ut\^t \in \bbR^d$\;
    Set $\displaystyle\tilde{\vec\ut}\^t \gets \vstack{-\langle \vec\ut\^t, \vec x\^t\rangle}{\vec\ut\^t}$\Comment*{\color{commentcolor}Lifting]\!\!\!\!\!} \label{line:lift}
    Set $\tilde{\Ut}\^{t+1} \gets \tilde{\Ut}\^t + \tilde{\vec\ut}\^t$
    }
\end{algorithm}

\subsection{Regret Analysis}\label{sec:regret}

In this section, we study the regret of $\algoshort$ under the idealized assumption that the optimization problem on \cref{line:oftrl} (\OFTRL step) is solved exactly at each time $t$. In \cref{sec:implementation} we will relax that assumption, and study the regret of \algoshort when the solution to \cref{line:oftrl} is approximated using variants of Newton's method.

To study the regret $\reg^T$ of \algoshort, defined in \eqref{eq:linreg}, it is useful to introduce the quantity $\tildereg^T$, which measures the regret incurred \emph{by the internal \OFTRL algorithm} (\cref{line:oftrl}) up to a time $T \in \N$ in the lifted space $\tilde \cX$, \emph{i.e.},
\[
    \tildereg^T \defeq \max_{(\lambda^*, \vec{y}^*)\in\tilde\cX} \sum_{t=1}^T \mleft\langle \tilde{\vec\ut}\^t, \vstack{\lambda^*}{\vec{y}^*} - \vstack{\lambda\^t}{\vec y\^t}\mright\rangle.
\]
As the following theorem clarifies, there is a strong connection between $\tildereg^T$ and $\reg^T$.
\begin{restatable}{theorem}{nonnegative}
    \label{theorem:nonnegative}
    For any time $T \in \N$ it holds that $\tildereg^T = \max\{0, \reg^T\}$. In particular, it follows that $\tildereg^T \ge 0$ and $\reg^T \le \tildereg^T$ for any $T \in \N$.
\end{restatable}

The nonnegativity of $\tildereg^T$ will be a crucial property in establishing \cref{theorem:traj}. Further, \cref{theorem:nonnegative} implies that a guarantee over the lifted space can be automatically translated to a regret bound over the original space $\cX$. Now let
\begin{equation}
    \label{eq:local_norm}
    % \begin{split}
        \pnt{\ \cdot\ } \defeq \|\cdot\|_{(\lambda\^t, \vec y\^t)}\quad\qquad\text{and}\quad\qquad
        %: \bbR\times\bbR^{d} \ni \vstack{a}{\vec z} &\mapsto \sqrt{\mleft(\frac{a}{\lambda\^t}\mright)^2 +\sum_{\ind=1}^d \mleft(\frac{\vec{z}[\ind]}{\vec{y}\^t[\ind]}\mright)^2}, \quad \text{and}
        %\\
        \dnt{\ \cdot\ }  \defeq \|\cdot\|_{*, (\lambda\^t, \vec y\^t)}
        %: \bbR\times\bbR^{d} \ni \vstack{a}{\vec z} &\mapsto \sqrt{\mleft(a \lambda\^t\mright)^2 +\sum_{\ind=1}^d \mleft(\vec{z}[\ind] \vec{y}\^t[\ind]\mright)^2}
    % \end{split}
\end{equation}
be the local norms centered at point $(\lambda\^t, \vec y\^t)$ produced by \OFTRL at time $t$ (\cref{line:oftrl}). In the next proposition we establish a refined RVU (Regret bounded by Variation in Utilities) bound in terms of this primal-dual norm pair. 
\begin{restatable}[RVU bound of \OFTRL in local norms]{proposition}{proplocalrvu}
    \label{thm:RVU}
    Let $\tildereg^T$ be the regret cumulated up to time $T$ by the internal \OFTRL algorithm. If $\|\vec{u}\^t \|_\infty\|\vec{x} \|_1 \leq 1$ at all times $t \in \range T$, then for any time horizon $T \in \N$ and learning rate $\eta \leq \etamin$,
    \[\tildereg^T \le 4 + \frac{(d+1)\log T}{\eta} + 5\eta\sum_{t=1}^T \dnt{ \tilde{\vec\ut}\^t - \tilde{\vec\ut}\^{t-1} }^2 - \frac{1}{27\eta}\sum_{t=1}^{T-1} \pnt{\vstack{\lambda\^{t+1}}{\vec y\^{t+1}} - \vstack{\lambda\^t}{\vec y\^t}}^2.
    \]
\end{restatable}

(We recall that $\tilu^{(0)} \defeq \vec{0}$.) \Cref{thm:RVU} differs from prior analogous results in that the regularizer is not a \emph{barrier} over the feasible set. Next, we show that the iterates produced by \OFTRL satisfy a refined notion of stability, which we refer to as \emph{multiplicative stability}.
\begin{restatable}[Multiplicative Stability]{proposition}{propftrlstab}
    \label{cor:ftrl stability}
    %(Assuming $\|\vec\ut\^t\|_\infty \le 1$) 
    For any time $t \in \N$ and learning rate $\eta \leq \etamin$, if $\|\vec{u}\^t \|_\infty\|\vec{x} \|_1 \leq 1$, 
    \[
        \pnt{\vstack{\lambda\^{t+1}}{\vec y\^{t+1}} - \vstack{\lambda\^t}{\vec y\^t}} \le 22\eta.
    \]
\end{restatable}
Intuitively, this property ensures that coordinates of successive iterates will have a small multiplicative deviation. We leverage this refined notion of stability to establish the following crucial lemma.
\begin{restatable}{lemma}{gamma}
    \label{lemma:gamma}
    For any time $t \in \N$ and learning rate $\eta \leq \etamin$, if $\|\vec{u}\^t \|_\infty\|\vec{x} \|_1 \leq 1$,
    \[
        \|\vec x\^{t+1} - \vec x\^t\|_1 \le 4 \|\cX\|_1 \pnt{\vstack{\lambda\^{t+1}}{\vec y\^{t+1}} - \vstack{\lambda\^t}{\vec y\^t}}.
    \]
\end{restatable}

Combining this lemma with \cref{thm:RVU} allows us to obtain an RVU bound for the original space $\cX$, with no dependencies on local norms.
% Next, we combine this lemma with \cref{thm:RVU} to obtain the following RVU bound, which is now cast in the original space $\cX$.
%
\begin{restatable}[RVU bound in the original (unlifted) space]{corollary}{rvuor}
    \label{cor:rvu}
    Fix any time $T \in \N$, and suppose that $\|\ut^{(t)}\|_\infty \leq B$ for any $t \in \range{T}$. If $\eta \leq \frac{1}{256 B \|\cX\|_1}$,
    \begin{equation*}
        \tildereg^T \leq 6 B \|\cX\|_1 + \frac{(d+1) \log T}{\eta} + 16 \eta \| \cX \|^2_1 \sum_{t=1}^{T-1} \|\vec{u}\^{t+1} - \vec{u}\^t \|^2_\infty - \frac{1}{512  \eta \|\cX\|_1^2} \sum_{t=1}^{T-1} \| \vec{x}\^{t+1} - \vec{x}\^t \|_1^2.
    \end{equation*}
\end{restatable}

\subsection{Main Result}

So far, in \Cref{sec:regret}, we have performed the analysis from the perspective of a single player, obtaining regret bounds that apply under an arbitrary sequence of utilities. Next, we assume that all players follow our dynamics such that the variation in one's utilities is now related to the variation in the joint strategies based on the smoothness condition of the utility function, connecting the last two terms of the RVU bound.
Further leveraging the nonnegativity of the regrets in the lifted space, we establish that the second-order path lengths of the dynamics up to time $T$ are bounded by $O(\log T)$:
\begin{restatable}{theorem}{traj}
    \label{theorem:traj}
    Suppose that \cref{assumption:smooth} holds for some parameters $B, L > 0$. If all players follow $\algoshort$ with learning rate $\eta \leq \min \left\{ \frac{1}{256B \|\cX\|_1}, \frac{1}{128 n L \|\cX\|_1^2} \right\}$, where $\|\cX\|_1 \defeq \max_{i \in \range{n}} \|\cX_i\|_1$, then
    \begin{equation}
        \label{eq:bounded-paths}
        \sum_{i=1}^n \sum_{t=1}^{T-1} \|\vec{x}_i\^{t+1} - \vec{x}_i\^t \|_1^2 \leq 6144 n \eta B \|\cX\|_1^3 + 1024 n (d+1) \|\cX\|_1^2 \log T.
    \end{equation}
\end{restatable}

Here we made the mild assumption that each player knows the values of $n$, $L$, $B$ and $\|\cX\|_1$ in order to appropriately tune the learning rate; otherwise, similar guarantees are possible via a standard application of the doubling trick. It is interesting to point out that~\eqref{eq:bounded-paths} holds even without the concavity condition (recall \Cref{assumption:smooth}). We next leverage \Cref{theorem:traj} to establish \cref{theorem:main}, the detailed version of which is given below. %We point out that, only for simplicity in the exposition, below we assume that $\|\cX\|_1, L \geq 1$.
\begin{restatable}[Detailed Version of \cref{theorem:main}]{theorem}{main}
    \label{theorem:main-detailed}
    Suppose that \cref{assumption:smooth} holds for some parameters $B, L > 0$. If all players follow \algoshort with learning rate $\eta = \min \left\{ \frac{1}{256 B \|\cX\|_1}, \frac{1}{128 n L \|\cX\|_1^2} \right\}$, then for any $T \in \N$ the regret $\reg_i^T$ of each player $i \in \range{n}$ can be bounded as
    \begin{equation}
        \label{eq:main-reg}
        \reg_i^T \leq 12 B \|\cX\|_1 + 256 (d+1) \max \left\{n L \|\cX\|_1^2, 2 B \|\cX\|_1 \right\} \log T.
    \end{equation}
    Furthermore, the algorithm can be adaptive so that if player $i$ is instead facing adversarial utilities, then $\reg_i^T = O(\sqrt{T})$.
\end{restatable}
%Recall that we have assumed that $\|\vec{u}^{(t)}_i\|_\infty \|\vec{x}_i \|_1 \leq 1$; instead, if we only knew that $\|\vec{u}^{(t)}_i\|_\infty \leq B$ (\Cref{assumption:smooth}), the aforementioned assumption can be met by rescaling the utilities by $1/(B \|\cX\|_1)$, introducing an additional multiplicative factor of $B \|\cX\|_1$ in the regret bounds (see \Cref{table:results}). 
For clarity, below we cast \eqref{eq:main-reg} of \Cref{theorem:main-detailed} in normal-form games with utilities normalized in the range $[-1, 1]$, in which case we can take $B = 1$, $L = 1$ and $\|\cX\|_1 = 1$.

\begin{corollary}[Normal-form Games]
    \label{cor:NFG}
    Suppose that all players in a normal-form game with $n \geq 2$ follow \algoshort with learning rate $\eta = \frac{1}{128 n}$. Then, for any $T \in \N$ and player $i \in \range{n}$,
    \begin{equation*}
        \reg_i^T \leq 12 + 256 n (d+1) \log T.
    \end{equation*}
\end{corollary}

\subsection{Implementation and Iteration Complexity}\label{sec:implementation}

% We have reduced obtaining near-optimal regret in general convex games to solving ``locally'' for each player \OFTRL optimization problem in \Cref{line:oftrl}. In other words, \Cref{theorem:main-detailed} posits an \OFTRL oracle. In this subsection we relax this assumption by considering more realistic oracles; such considerations are typically ignored in prior works.
In this section, we discuss the implementation and iteration complexity of \algoshort. The main difficulty in the implementation is the computation of the solution to the strictly concave \emph{nonsmooth} constrained optimization problem in \cref{line:oftrl}. 
We start by studying how the guarantees laid out in \cref{theorem:main-detailed} are affected when the exact solution to the \OFTRL problem (\cref{line:norm}) in \cref{algo:ours} is replaced with an approximation. Specifically, suppose that at all times $t$ the solution to the \OFTRL step (\cref{line:oftrl}) in \cref{algo:ours} is only \emph{approximately} solved within tolerance $\epsilon\^t$, in the sense that 
\[
\numberthis{eq:rel-error}
\mleft\|\vstack{\lambda\^t}{\vec y\^t} - \vstack{\lambda\^t_\star}{\vec y\^t_\star}\mright\|_{(\lambda_\star\^t, \vec y_\star\^t)} \le \epsilon\^t,
\]
where $(\lambda\^t, \vec y\^t) \in \bbR_{>0}^{d+1}$ and
    \[
    \vstack{\lambda\^t_\star}{\vec y\^t_\star}\defeq
    \argmax_{(\lambda, \vec y)\in\tilde\cX}\mleft\{ {\eta} \mleft\langle \tilde{\Ut}\^t + \tilde{\vec\ut}\^{t-1}, \vstack{\lambda}{\vec y} \mright\rangle + \log \lambda + \sum_{\ind=1}^d \log \vec{y}[\ind]\mright\}.
    \]

Then, it can be proven directly from the definition of regret that the guarantees given in \Cref{cor:rvu} deteriorate by an additive factor proportional to the sum of the tolerances $\sum_{t=1}^T \epsilon\^t$. %Specifically, this equivalent of \cref{cor:rvu} can be shown.
%
%\todo{Statement}
%
As an immediate corollary, when $\epsilon\^t \defeq \epsilon \defeq 1/T$, the conclusion of \cref{theorem:main-detailed} applies even when the solution to the optimization problem on \cref{line:oftrl} is only approximated up to $\epsilon$ tolerance.
Therefore, to complete our construction, it suffices to show that it is indeed possible to efficiently compute approximate solutions to the \OFTRL step (see \Cref{sec:approxiter}). In the remainder of this section, we show that this is indeed the case assuming access to two different type of oracles. It should be stressed that optimizing over a general convex set introduces several challenges not present under simplex domains, inevitably leading to an increased per-iteration complexity compared to algorithms designed specifically for normal-form games---such as OMWU.

%One of the main obstacles that arise when approximating the solution to the \OFTRL step, is that additive error guarantees on the approximation translate into large errors downstream, since the iterates produced by \algoshort are defined as the ratio $\vec y\^t / \lambda\^t$. In particular, when $\lambda\^t$ is small, additive errors can in principle translate into unbounded errors in the computation of $\vec x\^t$. Because of that technical hurdle, special care needs to be used.

\paragraph{Proximal Oracle} First, we will assume access to a \emph{proximal oracle} in local norm for the set $\tilde\cX$, that is, access to a function that is able to compute the solution to the (positive-definite) quadratic optimization problem
\begin{equation}
    \label{eq:quad_approx}
    \Pi_{\tilde{\vec w}}(\tilde{\vec g}) \defeq
    \argmin_{\tilde{\vec x} \,\in\, \tilde{\cX}} \mleft\{
    \tilde{\vec g}^\top \tilde{\vec x} + \frac{1}{2}\|\tilde{\vec x} - \tilde{\vec w} \|_{\tilde{\vec w}}^2
    \mright\}
    =
    \argmin_{\tilde{\vec x} \,\in\, \tilde{\cX}} \mleft\{
    \tilde{\vec g}^\top \tilde{\vec x} + \frac{1}{2}\sum_{\ind=1}^{d+1} \mleft(\frac{\tilde{\vec x}[\ind]}{\tilde{\vec w}[\ind]} - 1\mright)^2
    \mright\}
\end{equation}
for arbitrary centers $\tilde{\vec w} \in \bbR^{d+1}_{> 0}$ and gradients $\tilde{\vec g} \in \bbR^{d+1}$.
For certain sets $\cX \subseteq \bbR^d$, exact proximal oracles with polynomial complexity in the dimension $d$ can be given. In particular, we show that this is the case when $\cX$ is the strategy set of normal-form and extensive-form games by extending the approach of \citet[pp. 128-133]{Gilpin09:Algorithms}, as formalized below.

\begin{proposition}
Let $\cX \subseteq \bbR^d$ be the polytope of sequence-form strategies for a player in a perfect-recall extensive-form game. Then, the local proximal oracle $\Pi_{\tilde{\vec w}}(\tilde{\vec g})$ defined in \eqref{eq:quad_approx} can be implemented exactly in time polynomial in the dimension $d$ for any $\tilde{\vec w} \in \bbR_{> 0}^{d+1}$ and $\tilde{\vec g} \in \bbR^{d+1}$.
\end{proposition}

We provide the details and a more precise statement in \cref{app:proximal}. In this context, the following guarantee employs the \emph{proximal Newton algorithm} of~\citet{Dinh15:Composite}; see \Cref{algo:newton}.

\begin{theorem}[Proximal Newton]
    \label{theorem:prox-Newton}
    Given any $\epsilon > 0$, it is possible to compute $(\lambda\^t, \vec{y}\^t) \in \tilde{\cX}\cap\bbR_{>0}^{d+1}$ such that~\eqref{eq:rel-error} holds for $\epsilon\^t = \epsilon$ using $O(\log \log (1/\epsilon))$ operations and $O(\log \log (1/\epsilon))$ calls to the proximal oracle defined in~\Cref{eq:quad_approx}.
\end{theorem}
%As a result, using standard arguments, it is easy to see that $O(\log \log T)$ operations per iteration---assuming access to \eqref{eq:quad_approx}---suffice so that the implication of \Cref{theorem:main-detailed} extends.
%e remark that the refined relative-error guarantee of \eqref{eq:rel-error} is crucial as it implies multiplicative stability---in the sense of \Cref{cor:ftrl stability}---for the sequence of approximate iterates as well.

\paragraph{Linear Maximization Oracle} Moreover, we consider having access to a weaker \emph{linear maximization oracle (LMO)} for the set $\cX$:
\begin{equation}
    \label{eq:LMO}
    \mathcal{L}_\cX(\vec{u}) \defeq \argmax_{\vec{x} \in \cX} \langle \vec{x}, \vec{u} \rangle.
\end{equation}
Such an oracle is more realistic in many settings~\citep{Jaggi13:Revisiting}, and it is particularly natural in the context of games, where it can be thought of as a \emph{best response} oracle. We point out that an LMO for $\cX$ automatically implies an LMO for $\tilde{\cX}$. The following guarantee follows readily by applying the \emph{Frank-Wolfe (projected) Newton method}~\citep[Algotihms 1 and 2]{Liu20:A}.

\begin{theorem}[Frank-Wolfe Newton]
    \label{theorem:FW}
    Given any $\epsilon > 0$, it is possible to compute $(\lambda\^t, \vec y\^t) \in \tilde{\cX}\cap\bbR^{d+1}_{>0}$ such that \eqref{eq:rel-error} holds for $\epsilon\^t = \epsilon$ using $O(\poly (1/\epsilon))$ operations and $O(\poly (1/\epsilon))$ calls to the LMO oracle defined in~\Cref{eq:LMO}.
\end{theorem}

\paragraph{Experiments} Finally, while our main contribution is of theoretical nature, we also support our theory by conducting experiments on some standard extensive-form games (\Cref{section:experiments}). The experiments verify that under~\algoshort the regret of each player grows as $O(\log T)$.

% Again, by setting 
%\input{text/experiments}
\section{Conclusions}
\label{section:conclusions}

In this paper we developed $\algoshort$, a novel no-regret learning algorithm. We showed that when all players in a general convex game employ $\algoshort$, the regret of each player grows only as $O(\log T)$, thereby significantly extending and strengthening the scope of all prior work. Further, our uncoupled no-regret learning dynamics can be efficiently implemented using, for example, a proximal oracle for the underlying feasible set.

One caveat of our framework applied to the special case of normal-form games is that the dependence on the dimension is linear (\Cref{cor:NFG}) as opposed to logarithmic~\citep{Daskalakis21:Near}. Whether the entropic regularizer---which induces OMWU---can be incorporated into our framework is an important open question. Another interesting avenue for future research would be to explore having access to different types of oracles. For example, is it possible to extend \Cref{theorem:main-detailed} using a \emph{separation oracle} for the underlying set of strategies? If so, the ellipsoid algorithm~\citep{Bubeck15:Convex} would be the obvious candidate en route to implementing $\algoshort$.

\section*{Acknowledgments}

We are grateful to anonymous NeurIPS reviewers for many helpful comments. This material is based on work supported by the National Science Foundation under grants IIS-1718457, IIS-1901403, IIS-1943607, and CCF-1733556, and the ARO under award W911NF2010081. Christian Kroer is supported by the Office of Naval Research Young Investigator Program under grant N00014-22-1-2530.

\bibliography{refs}

\newpage
\appendix
\section{Omitted Proofs}
\label{app:proofs}

In this section we include all of the proofs omitted from the main body. For the convenience of the reader, we will restate each claim before proceeding with its proof. 

\subsection{Preliminary Proofs}

We commence with the proof of \Cref{proposition:unique}.

\propuniqoftrl*
\begin{proof}
    Uniqueness follow immediately from strict convexity. In the rest of the proof we focus on the existence part.
    
    We start by showing that there exists a point $\tilde x\in\tilde\cX$ whose coordinates are all strictly positive. By hypothesis (see \cref{sec:setup}), for every coordinate $\ind \in \range d$, there exists a point $\vec x_\ind$ such that $\vec x_\ind[\ind] > 0$. Hence, by convexity of $\cX \subseteq [0,+\infty)^d$ and by definition of $\tilde\cX$, the point
    \[
        (1, \vec x^\circ) \defeq 
        \mleft(1,~~ \frac{1}{d}\sum_{\ind=1}^d \vec x_\ind\mright).
    \]
    is such that $(1, \vec x^\circ) \in \tilde\cX \cap \bbR_{>0}^{d}$.
    
    Let now $M$ be the $\ell_\infty$ norm of the linear part in the \OFTRL step (\cref{line:oftrl} of \cref{algo:ours}). Then, a \emph{lower bound} on the optimal value $v^\star$ of objective is obtained by plugging in the point $(1,\vec x^\circ)$ at least
    \[
        v^\star &\ge -M(1 + \|\cX\|_1) + \sum_{\ind=1}^d \log \vec x^\circ[\ind].\numberthis{eq:vstar}
    \]
    Let now
    \[
        m \defeq \exp\mleft\{-(2M+d)(1 + \|\cX\|_1) + \sum_{\ind=1}^d \log \vec x^\circ[\ind]\mright\} > 0.\numberthis{eq:m}
    \]
    We will show that any point $(\lambda, \vec y) \notin [m,+\infty) \cap \tilde\cX$ cannot be optimal for the \OFTRL objective. Indeed, take a point $(\lambda, \vec y)\notin [m,+\infty) \cap \tilde\cX$. Then, at least one coordinate of $(\lambda, \vec y)$ is strictly less than $m$. If $\lambda < m$, then the objective value at $(\lambda, \vec y)$ is at most
    \[
        M\lambda + M\|\cX\|_1 + \log \lambda + \sum_{\ind=1}^d \log \vec y[\ind] &\le M (1 + \|\cX\|_1) + \log m + \sum_{\ind=1}^d \log\|\cX\|_1\\
            &\le M(1+\|\cX\|_1) + \log m + d (\|\cX\|_1 - 1)\\
            &< (M+d)(1+\|\cX\|_1) + \log m\\
            % &\le (M+d)(1+\|\cX\|_1) - (2M+d)(1 + \|\cX\|_1) + \sum_{\ind=1}^d \log \vec x^\circ[\ind]\\
            &= -M(1 + \|\cX\|_1) + \sum_{\ind=1}^d \log \vec x^\circ[\ind] &&\text{(from \eqref{eq:m})}\\
            &\le v^*, &&\text{(from \eqref{eq:vstar})}
    \]
    where the first inequality follows from upper bounding any coordinate of $\vec y$ with $\|\cX\|_1$, and the second inequality follows from using the inequality $\log z \le z - 1$, valid for all $z \in (0, +\infty)$.
    Similarly, if $\vec y[s] < m$ for some $s \in \range d$, then we can upper bound the objective value at $(\lambda, \vec y)$ as
    \[
        M + M\|\cX\|_1 + \log 1 + \sum_{\ind=1}^d \log \vec y[\ind] &\le M(1 + \|\cX\|_1) + \log m + \sum_{\ind=1}^d \log\|\cX\|_1\\
            &\le M(1+\|\cX\|_1) + (d-1) (\|\cX\|_1 - 1) + \log m\\
            &< (M+d)(1+\|\cX\|_1) + \log m \le v^*.
    \]
    So, in either case, we see that no optimal point can have any coordinate strictly less than $m$. Consequently, the maximizer of the \OFTRL step lies in the set $\mathcal{S} \defeq [m, +\infty)^{d+1} \cap \tilde{\cX}$. Since both $[m, +\infty)^{d+1}$ and $\tilde\cX$ are closed, and since $\tilde\cX$ is bounded by hypothesis, the set $\cal S$ is compact. Furthermore, note that $\cal S$ is nonempty, as $(1,\vec x^\circ) \in \cal S$, as for any $s \in \range d$
    \[
        \log m &= -(2M+d)(1 + \|\cX\|_1) + \sum_{\ind=1}^d \log \vec x^\circ[\ind]\\
            &\le -(2M+d)(1 + \|\cX\|_1) + \log \vec x^\circ[s] + (d-1) \log \|\cX\|_1 \\
            &\le -(2M+d)(1 + \|\cX\|_1) + \log \vec x^\circ[s] + (d-1)(\|\cX\|_1 - 1) \\
            &\le \log \vec x^\circ[s],
    \]
    implying that $(1,\vec x^\circ) \in [m, +\infty)^{d+1}$. Since $\cal S$ is compact and nonempty and the objective function is continuous, the optimization problem attains an optimal solution on $\cal S$ by virtue of Weierstrass' theorem.
\end{proof}

\nonnegative*
\begin{proof}
      First, by definition of $\tilde{\vec\ut}\^t$ in \cref{line:lift}, it follows that for any $t$,
      \[
            \left\langle \tilde{\vec\ut}\^t, \vstack{\lambda\^t}{\vec y\^t}\right\rangle =  \left\langle \tilde{\vec\ut}\^t, \vstack{1}{\vec x\^t}\right\rangle = 0.
      \]
      As a result, we have that $\max\{0, \reg^T\}$ is equal to
      \[
            &\max\mleft\{0, \max_{\vec{x}^* \in\cX} \sum_{t=1}^T \langle \vec\ut\^t, \vec{x}^* - \vec x\^t\rangle\mright\} %\\
            =
            \max\mleft\{0, \max_{\vec{x}^* \in\cX} \sum_{t=1}^T \mleft\langle \tilde{\vec\ut}\^t, \vstack{1}{\vec{x}^*} - \vstack{1}{\vec x\^t}\mright\rangle\mright\}\\
            =
            &\max\mleft\{0, \max_{\vec{x}^* \in \cX} \sum_{t=1}^T \mleft\langle \tilde{\vec\ut}\^t, \vstack{1}{\vec{x}^*}\mright\rangle\mright\}%\\
            =
            \max_{(\lambda^*, \vec{y}^* )\in\tilde\cX} \sum_{t=1}^T \mleft\langle \tilde{\vec\ut}\^t, \vstack{\lambda^*}{\vec{y}^*}\mright\rangle\\
            =
            &\max_{(\lambda^*, \vec{y}^*)\in\tilde\cX} \sum_{t=1}^T \mleft\langle \tilde{\vec\ut}\^t, \vstack{\lambda^*}{\vec{y}^* } - \vstack{\lambda\^t}{\vec y\^t}\mright\rangle
            =
            \tildereg^T,
      \]
      as we wanted to show.
\end{proof}

\subsection{Analysis of OFTRL with Logarithmic Regularizer}

For notational convenience, we define the log-regularizer $\psi: \tilde\cX \rightarrow \bbR_{\geq 0}$ as 
\[
    \ru{\pz} \defeq -\frac{1}{\eta}\sum_{\ind=1}^{d+1}\log\tilde{\vec{x}}[\ind],
\]
and its induced Bregman divergence
\[
      \bg{\pz}{\tilde{\vec z}} \defeq \frac{1}{\eta}\sum^{d+1}_{\ind=1}h\left(\frac{\tilde{\vec{x}}[\ind]}{\tilde{\vec z}[\ind]}\right),
      \quad\text{where}\; h(a)=a-1-\ln(a).
\]
Moreover, we define
\begin{equation}\label{eq:def_p}
      \p{t} = \argmax_{\pz\in\tilde\cX} -F_t(\pz)= \argmin_{\pz\in\tilde\cX} F_t(\pz), \quad\text{where}\; F_t(\pz) = -\ir{\Ls{t}+\ls{t-1},\pz}+\ru{\pz}.
\end{equation}
We note that $F_t$ is a convex function for each $t$ and $\p{t}$ is exactly equal to $\vstack{\lambda\^t}{\vec y\^t}$ computed by Algorithm~\ref{algo:ours}.
Further,  we define an auxiliary sequence $\{\pp{t}\}_{t=1,2,\ldots}$ defined as follows.
\begin{equation}\label{eq:def_pp}
      \pp{t} = \argmax_{\pz\in\tilde\cX} -G_{t}(\pz)=\argmin_{\pz\in\tilde\cX} G_{t}(\pz), \quad\text{where}\; G_t(\pz) = -\ir{\Ls{t},\pz}+\ru{\pz}.
\end{equation}
Similarly, $G_t$ is a convex function for each $t$.
We also recall the primal and dual norm notation:
\[
      \pnt{\ppz}&=\sum^{d+1}_{\ind=1}\left(\frac{\tilde{\vec{z}}[\ind]}{\tilde{\vec{x}}\^t[\ind]}\right)^2, \quad \dnt{\tilde{\vec z}}=\sum^{d+1}_{\ind=1}\left(\tilde{\vec{x}}\^t[\ind]\tilde{\vec z}[\ind]\right)^2.
\]
Finally, for a $(d+1) \times (d+1)$ positive definite matrix $\mat{M}$, we use $\|\tilde{\vec z}\|_\mat{M}$ to denote the induced quadratic norm $\sqrt{\tilde{\vec z}^\top \mat{M} \tilde{\vec z}}$. We are now ready to establish \cref{thm:RVU}.

\proplocalrvu*
\begin{proof}[Proof of \cref{thm:RVU}]
      For any comparator $\cp \in \tilde\cX$, define $\cp'=\frac{T-1}{T}\cdot \cp+\frac{1}{T}\cdot\p{1} \in \tilde\cX$, where we recall $\p{1} = \argmin_{\pz\in\tilde\cX} F_1(\pz) = \argmin_{\pz\in\tilde\cX} \ru{\pz}$. Then, we have
      \[
            \sum_{t=1}^{T}\ir{\cp-\p{t},\ls{t}} &= \sum_{t=1}^{T}\ir{\cp-\cp',\ls{t}}+\sum_{t=1}^{T}\ir{\cp'-\p{t},\ls{t}} \\
            &= \frac{1}{T}\sum_{t=1}^T \ir{\cp-\p{1}, \ls{t}} +\sum_{t=1}^{T}\ir{\cp'-\p{t},\ls{t}}
            \\
            &\le 4 +\sum_{t=1}^{T}\ir{\cp'-\p{t},\ls{t}},
      \]
      where the last inequality follows from Cauchy-Schwarz together with the assumption that $\|{\vec\ut}\^t\|_\infty \le \frac{1}{\| \cX \|_1}$.
      
      Now, by standard Optimistic FTRL analysis (see Lemma~\ref{lem:oftrl regret}), the last term $\sum_{t=1}^{T}\ir{\cp'-\p{t},\ls{t}}$ (cumulative regret against $\cp'$) is bounded by
      \[
      \sum_{t=1}^{T}\ir{\cp'-\p{t},\ls{t}} &\le
            \ru{\cp'}-\ru{\p{1}}+\sum^T_{t=1}\ir{\pp{t+1}-\p{t},\ls{t}-\ls{t-1}} \\
            &\hspace{3.4cm}-\sum^T_{t=1}\left(\bg{\p{t}}{\pp{t}}+\bg{\pp{t+1}}{\p{t}}\right).
      \]
      For the term $\ru{\cp'}-\ru{\p{1}}$, a direct calculation using definitions shows
      \[
            \ru{\cp'}-\ru{\p{1}} = \frac{1}{\eta}\sum_{i=1}^{d+1} \log \frac{\tilde{\vec x}^{(1)}[i]}{\cp'[i]} \leq \frac{d+1}{\eta}\log T.
      \]
      For the other terms, we apply Lemma~\ref{lem:stab} and Lemma~\ref{lem:neg}, which completes the proof.
\end{proof}

\begin{lemma}\label{lem:oftrl regret}
      The update rule~\eqref{eq:def_p} ensures the following for any $\cp \in \tilde\cX$:
      \[
            \sum^T_{t=1}\ir{\cp-\p{t},\ls{t}} &\le \ru{\cp}- \ru{\p{1}}+\sum^T_{t=1}\ir{\pp{t+1}-\p{t},\ls{t}-\ls{t-1}} \\
            &\quad\quad -\sum^T_{t=1}\left(\bg{\p{t}}{\pp{t}}+\bg{\pp{t+1}}{\p{t}}\right).
      \]
\end{lemma}
\begin{proof}
      First note that for any convex function $F: \tilde\cX \rightarrow \bbR$ and a minimizer $\pz^\star$, we have for any $\pz \in \tilde\cX$:
      \[
            F(\pz^\star) = F(\pz) - \ir{\nabla F(\pz^\star),  \pz - \pz^\star} - D_F(\pz, \pz^\star)
            \leq F(\pz)- D_F(\pz, \pz^\star),
      \]
      where $D_F$ is the Bregman Divergence induced by $F$ and the inequality is by the first-order optimality.
      Using this fact and the optimality of $\pp{t}$, we have
      \[
            G_t(\pp{t}) &\le G_t(\p{t}) -\bg{\p{t}}{\pp{t}}\\
            &=F_t(\p{t})+\ir{\p{t},\ls{t-1}}-\bg{\p{t}}{\pp{t}}
      \]
      Similarly, using the optimality of $\p{t}$, we have
      \[
            F_t(\p{t}) &\le F_t(\pp{t+1}) -\bg{\pp{t+1}}{\p{t}}\\
            &= G_{t+1}(\pp{t+1})+\ir{\pp{t+1},\ls{t}-\ls{t-1}}-\bg{\pp{t+1}}{\p{t}}
      \]
      Combining the inequalities and summing over $t$, we have
      \[
            G_1(\pp{1}) &\le G_{T+1}(\pp{T+1}) +\sum^T_{t=1}\left(\ir{\p{t},\ls{t}}+\ir{\pp{t+1}-\p{t},\ls{t}-\ls{t-1}}\right)\\
            &\hspace{5cm}+\sum^T_{t=1}\left(-\bg{\p{t}}{\pp{t}}-\bg{\pp{t+1}}{\p{t}}\right).
      \]
      Observe that $G_1(\pp{1})=\ru{\p{1}}$ and $G_{T+1}(\pp{T+1}) \le -\ir{\cp,\Ls{T+1}}+\ru{\cp}$. Rearranging then proves the lemma.
\end{proof}

\begin{lemma}\label{lem:stab}
      If $\eta\le \etamin$, then we have
      \begin{align}
            \pnt{\pp{t+1}-\p{t}} & \le 5\eta\dnt{\ls{t}-\ls{t-1}} \le 10\sqrt{2}\eta \le 15\eta, \label{eq:stab1} \\
            \pnt{\p{t+1}-\p{t}}  & \le 5\eta\dnt{2\ls{t}-\ls{t-1}} \le 15\sqrt{2}\eta \le 22\eta. \label{eq:stab2}
      \end{align}
      \begin{proof}
            The second part of both inequalities is clear by definitions:
            \[
                  \dnt{\ls{t}-\ls{t-1}}^2 &= \left(\lambda\^t\left(\ir{\vec x\^t, \vec \ut\^t} - \ir{\vec x\^{t-1}, \vec \ut\^{t-1}}\right)\right)^2
                  + \sum_{\ind=1}^d \left(\vec y\^t[\ind] \left(\ut\^t[\ind] - \ut\^{t-1}[\ind]\right) \right)^2 \\
                  &\leq 4(\lambda\^t)^2 + \frac{4}{\|\cX\|_1^2}\sum_{\ind=1}^d \mleft(\vec y\^t[\ind]\mright)^2 \leq 8,
            \]
            where we use $\ir{\vec x\^\tau, \vec \ut\^\tau}\le \|\vec{x}\^\tau \|_1 \|\vec{u}\^\tau \|_\infty \leq 1$ and $|\vec{u}\^\tau[r]|\le\frac{1}{\|\cX\|_1}$ for any time $\tau$ and any coordinate $r$ by the assumption, and similarly,
            \[
                  \dnt{2\ls{t}-\ls{t-1}}^2 &= \left(\lambda\^t\left(2\ir{\vec x\^t, \vec \ut\^t} - \ir{\vec x\^{t-1}, \vec \ut\^{t-1}}\right)\right)^2
                  + \sum_{\ind=1}^d \left(\vec{y}\^t[\ind] \left(2\ut\^t[\ind] - \ut\^{t-1}[\ind]\right) \right)^2 \\
                  &\leq 9(\lambda\^t)^2 + \frac{9}{\|\cX\|_1^2}\sum_{\ind=1}^d \mleft(\vec y\^t[\ind]\mright)^2 \leq 18
            \]
            %\todo{Gabri: more details about where the bounds come from}
            To prove the first inequality in Eq.~\eqref{eq:stab1},
            let $\mathcal{E}_t = \mleft\{\pz: \pnt{\pz-\p{t}}\le 5\eta\dnt{\ls{t}-\ls{t-1}}\mright\}$.
            Noticing that $\pp{t+1}$ is the minimizer of the convex function $G_{t+1}$,
            to show $\pp{t+1} \in \mathcal{E}_t$, it suffices to show that for all $\cp$ on the boundary of $\mathcal{E}_t$, we have $G_{t+1}(\cp)\ge G_{t+1}(\p{t})$.
            Indeed,
            using Taylor's theorem, for any such $\cp$, there is a point $\vec\xi$ on the line segment between $\p{t}$ and $\cp$ such that
            \begin{align*}
                  G_{t+1}(\cp) & =G_{t+1}(\p{t})+\ir{\nabla G_{t+1}(\p{t}),\cp-\p{t}}+\frac{1}{2}\nms{\cp-\p{t}}{\nabla^2G_{t+1}(\vec\xi)}                     \\
                               & =G_{t+1}(\p{t})-\ir{\ls{t}-\ls{t-1},\cp-\p{t}}+\ir{\nabla F_t(\p{t}),\cp-\p{t}}+\frac{1}{2}\pnms{\cp-\p{t}}{\vec\xi}          \\
                               & \ge G_{t+1}(\p{t})-\ir{\ls{t}-\ls{t-1},\cp-\p{t}}+\frac{1}{2}\pnms{\cp-\p{t}}{\vec\xi}\tag{by the optimality of $\p{t}$}      \\
                               & \ge G_{t+1}(\p{t})-\dnt{\ls{t}-\ls{t-1}}\pnt{\cp-\p{t}} +\frac{1}{2}\pnms{\cp-\p{t}}{\vec\xi}. \tag{by H\"older's inequality} \\
                               & \ge G_{t+1}(\p{t})-\dnt{\ls{t}-\ls{t-1}}\pnt{\cp-\p{t}} +\frac{2}{9\eta}\pnt{\cp-\p{t}}^2 \tag{$\star$}                       \\
                               & =G_{t+1}(\p{t}) + \frac{5}{9}\eta\dnt{\ls{t}-\ls{t-1}}^2  \tag{$\pnt{\cp-\p{t}}= 5\eta\dnt{\ls{t}-\ls{t-1}}$}                 \\
                               & \ge G_{t+1}(\p{t}).
            \end{align*}
            Here, the inequality $(\star)$ holds because Lemma~\ref{lem:multi_stab} (together with the condition $\eta \le \etamin$) shows $\frac{1}{2}\tilde{\vec x}\^t[i] \leq \cp[i] \leq \frac{3}{2}\tilde{\vec x}\^t[i]$, which implies $\frac{1}{2}\tilde{\vec x}\^t[i] \leq \vec\xi[i] \leq \frac{3}{2}\tilde{\vec x}\^t[i]$ as well, and thus $\nabla^2 \ru{\vec\xi}\succeq \frac{4}{9}\nabla^2 \ru{\p{t}}$.
            This finishes the proof for Eq.~\eqref{eq:stab1}.
            The first inequality of Eq.~\eqref{eq:stab2} can be proven in the same manner.
      \end{proof}
\end{lemma}

\propftrlstab*
\begin{proof}
    The statement is proved in \cref{lem:stab}.
\end{proof}

\begin{lemma}\label{lem:multi_stab}
      If $\cp$ satisfies $\|\cp - \p{t}\|_t \leq \frac{1}{2}$, then $\frac{1}{2}\tilde{\vec x}\^t[i] \leq \cp[i] \leq \frac{3}{2}\tilde{\vec x}\^t[i]$ for every coordinate $i$.
      %and also $\frac{1}{\eta}\pnt{\pz}^2 \leq  \pnms{\pz}{\cp} \leq \frac{1}{\eta}\pnt{\pz}^2$ for any $\pz \in \tilde \cX$.
\end{lemma}
\begin{proof}
      By definition,  $\|\cp - \p{t}\|_t \leq \frac{1}{2}$ implies for any $i$, $\frac{\left| \cp[i]-\tilde{\vec x}\^t[i]\right|}{\tilde{\vec x}\^t[i]}\le \frac{1}{2}$, and thus $\frac{1}{2}\tilde{\vec x}\^t[i] \leq \cp[i] \leq \frac{3}{2}\tilde{\vec x}\^t[i]$.
\end{proof}

\begin{lemma}\label{lem:neg}
      If $\eta \leq \frac{1}{50}$, then we have
      \begin{align*}
            \sum^T_{t=1}\left(\bg{\p{t}}{\pp{t}}+\bg{\pp{t+1}}{\p{t}}\right)\ge\frac{1}{27\eta}\sum^{T-1}_{t=1}\pnt{\p{t+1}-\p{t}}^2.
      \end{align*}
\end{lemma}
\begin{proof}
      Recall $h(a)=a-1-\ln(a)$ and  $\bg{\pz}{\tilde{\vec z}}=\frac{1}{\eta}\sum^{d+1}_{i=1}h\left(\frac{\tilde{\vec x}[i]}{ \tilde{\vec z}[i]}\right)$. We proceed as
      \begin{align*}
            \mathrlap{\sum^T_{t=1}\left(\bg{\p{t}}{\pp{t}}+\bg{\pp{t+1}}{\p{t}}\right)}\hspace{2cm}                                                                                                                                                                                            \\
             & \ge\sum^{T-1}_{t=1}\left(\bg{\p{t+1}}{\pp{t+1}}+\bg{\pp{t+1}}{\p{t}}\right)                                                                                                                                                                                                     \\
             & =\frac{1}{\eta}\sum^{T-1}_{t=1}\sum^{d+1}_{i=1}\left(h\left(\frac{\tilde{\vec x}\^{t+1}[i]}{\pp{t+1}[i]}\right)+h\left(\frac{\pp{t+1}[i]}{\tilde{\vec x}\^{t}[i]}\right)\right)                                                                                                 \\
             & \ge\frac{1}{6\eta}\sum^{T-1}_{t=1}\sum^{d+1}_{i=1}\left(\frac{(\tilde{\vec x}\^{t+1}[i]-\pp{t+1}[i])^2}{\left(\pp{t+1}[i]\right)^2}+\frac{(\pp{t+1}[i]-\tilde{\vec x}\^{t}[i])^2}{(\tilde{\vec x}\^{t}[i])^2}\right)\tag{$h(y)\ge\frac{(y-1)^2}{6}$ for $y\in[\frac{1}{3}, 3]$} \\
             & \ge\frac{2}{27\eta}\sum^{T-1}_{t=1}\sum^{d+1}_{i=1}\left(\frac{(\tilde{\vec x}\^{t+1}[i]-\pp{t+1}[i])^2}{\left(\tilde{\vec x}\^{t}[i]\right)^2}+\frac{(\pp{t+1}[i]-\tilde{\vec x}\^{t}[i])^2}{(\tilde{\vec x}\^{t}[i])^2}\right)                                                \\
             & \ge\frac{1}{27\eta}\sum^{T-1}_{t=1}\sum^{d+1}_{i=1}\left(\frac{(\tilde{\vec x}\^{t+1}[i]-\tilde{\vec x}\^{t}[i])^2}{\left(\tilde{\vec x}\^{t}[i]\right)^2}\right)=\frac{1}{27\eta}\sum^{T-1}_{t=1}\pnt{\p{t+1}-\p{t}}^2.
      \end{align*}

      Here, the second and the third inequality hold because by Lemma~\ref{lem:stab} and Lemma~\ref{lem:multi_stab}, we have $\frac{1}{2} \leq \frac{\tilde{\vec z}\^{t+1}[i]}{\tilde{\vec x}\^{t}[i]} \leq \frac{3}{2}$ and $\frac{1}{2} \leq \frac{\tilde{\vec x}\^{t+1}[i]}{\tilde{\vec x}\^{t}[i]} \leq \frac{3}{2}$, and thus $\frac{1}{3} \leq \frac{\tilde{\vec x}\^{t+1}[i]}{\tilde{\vec z}\^{t+1}[i]} \leq 3$.
\end{proof}

\subsection{RVU Bound in the Original Space}

Next, we establish an RVU bound in the original (unlifted) space, namely \Cref{cor:rvu}. To this end, we first proceed with the proof of \cref{lemma:gamma}, which boils down to the following simple claim.

\begin{lemma}\label{lem:ratio close}
    Let $(\lambda, \vec y), (\lambda',\vec y') \in \tilde\cX\cap\bbR^{d+1}_{> 0}$ be arbitrary points such that
    \[
        \mleft\|\vstack{\lambda'}{\vec y'} - \vstack{\lambda}{\vec y} \mright\|_{(\lambda, \vec y)} \le \frac{1}{2}.
    \]
    Then,
    \[
        \mleft\| \frac{\vec y}{\lambda} - \frac{\vec y'}{\lambda'} \mright\|_1 \le 4 \|\cX\|_1 \cdot \mleft\|\vstack{\lambda'}{\vec y'} - \vstack{\lambda}{\vec y} \mright\|_{(\lambda,\vec y)}.
    \]
\end{lemma}
\begin{proof}
      Let $\mul$ be defined as
      \begin{equation}
            \label{eq:def-mul}
            \mul \defeq \max \left\{ \left| \frac{\lambda'}{\lambda} - 1 \right|, \max_{\ind \in \range{d}} \left| \frac{\vec{y}'[\ind]}{\vec{y}[\ind]} - 1 \right| \right\}.
      \end{equation}
      By definition,
      \begin{equation*}
            \left| \frac{\lambda'}{\lambda} - 1 \right| \leq \mul,
      \end{equation*}
      which in turn implies that
      \begin{equation}
            \label{eq:lam-stab}
            (1 - \mul ) \lambda \leq \lambda' \leq (1 + \mul ) \lambda.
      \end{equation}
      Similarly, for any $\ind \in \range{d}$,
      \begin{equation}
            \label{eq:y-stab}
            (1 - \mul) \vec{y}[\ind] \leq \vec{y}'[\ind] \leq (1 + \mul) \vec{y}[\ind].
      \end{equation}
      As a result, combining \eqref{eq:lam-stab} and \eqref{eq:y-stab} we get that for any $\ind \in \range{d}$,
      \begin{equation*}
            \frac{\vec{y}'[\ind]}{\lambda'} - \frac{\vec{y}[\ind]}{\lambda} \leq \left( \frac{1 + \mul }{1 - \mul } - 1 \right) \frac{\vec{y}[\ind]}{\lambda} \leq 4 \mul \frac{\vec{y}[\ind]}{\lambda} =  4 \mul \vec{x}[\ind],
      \end{equation*}
      since $\mul \leq \frac{1}{2}$. Similarly, by \eqref{eq:lam-stab} and \eqref{eq:y-stab},
      \begin{equation*}
            \frac{\vec{y}[\ind]}{\lambda} - \frac{\vec{y}'[\ind]}{\lambda'} \leq \left( 1 - \frac{1 - \mul }{1 + \mul } \right) \frac{\vec{y}[\ind]}{\lambda} \leq 2 \mul \frac{\vec{y}[\ind]}{\lambda} = 2 \mul \vec{x}[\ind].
      \end{equation*}
      Thus, it follows that
      \begin{equation*}
            \left| \frac{\vec{y}'[\ind]}{\lambda'} - \frac{\vec{y}[\ind]}{\lambda} \right| \leq 4 \mul \vec{x}[\ind],
      \end{equation*}
      in turn implying that
      \begin{equation}
            \label{eq:mul-bound}
            \| \vec{x}' - \vec{x} \|_1 = \sum_{\ind=1}^d \left| \frac{\vec{y}'[\ind]}{\lambda'} - \frac{\vec{y}[\ind]}{\lambda} \right| \leq 4 \mul \sum_{\ind=1}^d \vec{x}[\ind] \leq 4 \mul \| \cX\|_1.
      \end{equation}
      Moreover, by definition of \eqref{eq:def-mul},
      \begin{equation*}
            \left( \mul \right)^2 \leq \pnt{\vstack{\lambda'}{\vec y'} - \vstack{\lambda}{\vec y}}^2.
      \end{equation*}
      Finally, combining this bound with \eqref{eq:mul-bound} concludes the proof.
\end{proof}

\gamma*
\begin{proof}
      Since $\eta \le \frac{1}{50}$ by assumption, we have
      \[
        \pnt{\vx\^{t+1} - \vx\^t} \le 22\eta < \frac{1}{2}.
      \]
      Hence, we are in the domain of applicability of \cref{lem:ratio close}, which immediately yields the statement.
\end{proof}

\rvuor*

\begin{proof}
      At first, assume that $\|\vec{\ut}\^t\|_\infty \le 1/\|\cX\|_1$. By definition of the induced dual local norm in \eqref{eq:local_norm},
      \begin{align}
            \|\tilu\^t - \tilu\^{t-1} \|^2_{*, t} & \leq (\langle \vec{x}\^t, \vec{u}\^t \rangle - \langle \vec{x}\^{t-1}, \vec{u}\^{t-1} \rangle)^2 (\lambda\^t)^2 + \sum_{\ind=1}^d (\vec{y}[\ind])^2 (\vec{u}\^t[\ind] - \vec{u}\^{t-1}[\ind] )^2 \notag \\
                                                  & \leq (\langle \vec{x}\^t, \vec{u}\^t \rangle - \langle \vec{x}\^{t-1}, \vec{u}\^{t-1} \rangle)^2 + \sum_{\ind=1}^d (\vec{x}[\ind])^2 (\vec{u}\^t[\ind] - \vec{u}\^{t-1}[\ind] )^2 \notag                \\
                                                  & \leq \left( \langle \vec{x}\^t, \vec{u}\^t \rangle - \langle \vec{x}\^{t-1}, \vec{u}\^{t-1} \rangle \right)^2 + \|\cX\|_1^2 \| \vec{u}\^t - \vec{u}\^{t-1} \|^2_\infty, \label{align:mis-util}
      \end{align}
      for any $t \geq 2$. Further, by Young's inequality,
      \begin{align*}
            \left( \langle \vec{x}\^t, \vec{u}\^t \rangle - \langle \vec{x}\^{t-1}, \vec{u}\^{t-1} \rangle \right)^2 & \leq 2 \left( \langle \vec{x}\^t, \vec{u}\^t - \vec{u}\^{t-1} \rangle \right)^2 + 2 \left( \langle \vec{x}\^t - \vec{x}\^{t-1}, \vec{u}\^{t-1} \rangle \right)^2 \\
                                                                                                                     & \leq 2 \|\cX\|_1^2 \|\vec{u}\^t - \vec{u}\^{t-1} \|^2_\infty + \frac{2}{\|\cX\|_1^2}  \|\vec{x}\^t - \vec{x}\^{t-1} \|_1^2.
      \end{align*}
      Combining with \eqref{align:mis-util},
      \begin{equation*}
            \|\tilu\^t - \tilu\^{t-1} \|^2_{*, t} \leq 3 \|\cX\|_1^2 \|\vec{u}\^t - \vec{u}\^{t-1} \|^2_\infty + \frac{2}{\|\cX\|_1^2} \|\vec{x}\^t - \vec{x}\^{t-1} \|_1^2,
      \end{equation*}
      for $t \geq 2$, since $\|\vec{u}\|_\infty \leq \frac{1}{\|\cX\|_1} $ (by assumption). Further, $\|\tilu^{(1)} - \tilu^{(0)} \|^2_{*, t} = \|\tilu^{(1)} \|^2_{*, t} \leq 2$.
      Combining with \cref{thm:RVU} and \cref{lemma:gamma}, we get that $\tildereg^T$ is upper bounded by
      \begin{align*}
            6 + \frac{(d+1) \log T}{\eta} + 16 \eta \|\cX\|_1^2 \sum_{t=1}^{T-1} \|\vec{u}\^{t+1} - \vec{u}\^t \|_\infty^2 + \frac{1}{\|\cX\|_1^2} \left( 10 \eta - \frac{1}{432 \eta}  \right) \sum_{t=1}^{T-1} \|\vec{x}\^{t+1} - \vec{x}\^t \|_1^2 \\
            \leq 6 + \frac{(d+1) \log T}{\eta} + 16 \eta \|\cX\|_1^2 \sum_{t=1}^{T-1} \|\vec{u}\^{t+1} - \vec{u}\^t \|_\infty^2 - \frac{1}{512 \eta \|\cX\|_1^2} \sum_{t=1}^{T-1} \|\vec{x}\^{t+1} - \vec{x}\^t \|_1^2.
      \end{align*}

      Finally, we relax the assumption that $\|\vec{\ut}\^t\|_\infty \le 1/\|\cX\|_1$. In that case, one can reduce to the above analysis by first rescaling all utilities by the factor $1/(B\|\cX\|_1)$---which in turn is equivalent to rescaling the learning rate $\eta$ by $1/(B\|\cX\|_1)$. We then need to correct for the fact that the norm of the difference of utilities gets rescaled by a factor $1/(B\|\cX\|_1)^2$, and that the regret $\tildereg^T$ with respect to the original utilities is a factor $B\|\cX\|_1$ larger than the regret measured on the rescaled utilities. Taking these considerations into account leads to the statement.
\end{proof}

\subsection{Main Result: Proof of \texorpdfstring{\Cref{theorem:main-detailed}}{Theorem 4}}

Finally, we are ready to establish \Cref{theorem:main-detailed}. To this end, the main ingredient is the bound on the second-order path lengths predicted by \Cref{theorem:traj}, which is recalled below.

\traj*

\begin{proof}
      By \cref{assumption:smooth}, it follows that for any player $i \in \range{n}$,
      \begin{equation*}
            \left( \| \vec{u}_i\^{t+1} - \vec{u}_i\^t \|_\infty \right)^2 \leq \left( L \sum_{j = 1}^n \|\vec{x}_j\^{t+1} - \vec{x}_j\^t \|_1 \right)^2 \leq L^2 n \sum_{j=1}^n \|\vec{x}_j\^{t+1} - \vec{x}_j\^t \|_1^2,
      \end{equation*}
      by Jensen's inequality. Hence, by \cref{cor:rvu} the regret $\reg_i^T$ of each player $i \in \range{n}$ can be upper bounded by
      \begin{equation*}
            6 B \|\cX\|_1 + \frac{(d+1) \log T}{\eta} + 16 \eta \|\cX\|_1^2 L^2 n \sum_{j=1}^n \sum_{t=1}^{T-1} \|\vec{x}_j\^{t+1} - \vec{x}_j\^t \|_1^2  - \frac{1}{512 \eta \|\cX \|_1^2} \sum_{t=1}^{T-1} \|\vec{x}_i\^{t+1} - \vec{x}_i\^t \|_1^2,
      \end{equation*}
      Summing over all players $i \in \range{n}$, we have that
      \begin{align*}
            \sum_{i=1}^n \tildereg_i^T & \leq 6 n B \|\cX\|_1 + n \frac{(d+1) \log T}{\eta} +  \sum_{i=1}^n \left( 16 \eta \|\cX\|_1^2 L^2 n^2 - \frac{1}{512 \eta \|\cX\|_1^2} \right) \sum_{t=1}^{T-1} \|\vec{x}_i\^{t+1} - \vec{x}_i\^t \|_1^2 \\
                                       & \leq 6 n B \|\cX\|_1 + n \frac{(d+1) \log T}{\eta} - \frac{1}{1024 \eta \|\cX\|_1^2} \sum_{i=1}^n \sum_{t=1}^{T-1} \|\vec{x}_i\^{t+1} - \vec{x}_i\^t \|_1^2,
      \end{align*}
      since $\eta \leq \frac{1}{256 n L \|\cX\|_1^2}$. Finally, the theorem follows since $\sum_{i=1}^n \tildereg_i^T \geq 0$, which in turn follows directly from \cref{theorem:nonnegative}.
\end{proof}

\main*

\begin{proof}
      First of all, by \cref{assumption:smooth} we have that for any player $i \in \range{n}$,
      \begin{equation*}
            \|\vec{u}_i\^{t+1} - \vec{u}_i\^t\|^2_\infty \leq \left( L \sum_{j=1}^n \|\vec{x}_j\^{t+1} - \vec{x}_j\^t \|_1 \right)^2 \leq L^2 n \sum_{j=1}^n \|\vec{x}_j\^{t+1} - \vec{x}_j\^t \|_1^2.
      \end{equation*}
      Hence, summing over all $t$,
      \begin{align*}
            \sum_{t=1}^{T-1} \|\vec{u}_i\^{t+1} - \vec{u}_i\^t \|_\infty^2 & \leq L^2 n \sum_{t=1}^{T-1} \sum_{j=1}^n \|\vec{x}_j\^{t+1} - \vec{x}_j\^t \|_1^2 \\
                                                                           & \leq 6144 n^2 L^2 \eta B \|\cX\|_1^3 + 1024 n^2 L^2 (d+1) \|\cX\|_1^2 \log T,
      \end{align*}
      where the last bound uses \cref{theorem:traj}. As a result, from \cref{cor:rvu}, if $\eta = \frac{1}{128 n L \|\cX\|_1^2}$,
      \begin{align*}
            \tildereg_i^T & \leq 6 B \|\cX\|_1 + \frac{(d+1) \log T}{\eta} + 16 \eta \|\cX\|_1^2 \sum_{t=1}^{T-1} \|\vec{u}_i\^{t+1} - \vec{u}_i\^t \|^2_\infty \\
                          & \leq 12 B \|\cX\|_1 + 256 (d+1) n L \|\cX\|_1^2 \log T.
      \end{align*}
      Thus, the bound on $\reg_i^T$ follows directly since $\reg_i^T \leq \tildereg_i^T$ by \cref{theorem:nonnegative}. The case where $\eta = \frac{1}{256 B \|\cX\|_1}$ is analogous. 
      
      Next, let us focus on the adversarial bound. Each player can simply check whether there exists a time $t \in \range{T}$ such that
      \begin{equation}
            \label{eq:check}
            \sum_{\tau=1}^{t-1} \|\vec{u}_i^{(\tau+1)} - \vec{u}_i^{(\tau)} \|_\infty^2 > 6144 n^2 L^2 \eta B \|\cX\|_1^3 + 1024 n^2 L^2 (d+1) \|\cX\|_1^2 \log t.
      \end{equation}
      In particular, we know from \cref{theorem:traj} that when all players follow the prescribed protocol \eqref{eq:check} will never by satisfied. On the other hand, if there exists time $t$ so that \eqref{eq:check} holds, then it suffices to switch to any no-regret learning algorithm tuned to face adversarial utilities.
\end{proof}

\subsection{Extending the Analysis under Approximate Iterates}
\label{sec:approxiter}

In this subsection we describe how to extend our analysis, and in particular \Cref{theorem:main-detailed}, when the \OFTRL step of \Cref{algo:ours} at time $t$ is only computed with tolerance $\epsilon^{(t)}$, in the sense of \eqref{eq:rel-error}. We start by extending \Cref{thm:RVU} below.

\begin{proposition}[Extension of \Cref{thm:RVU}]
    \label{prop:approx-rvu}
Let $\tildereg^T$ be the regret cumulated up to time $T$ by the internal \OFTRL algorithm producing approximate iterates $(\lambda^{(t)}, \vec{y}^{(t)}) \in \tilde{\cX}$, for any $t \in \range{T}$. Then, for any $T \in \N$ and learning rate $\eta \leq \frac{1}{50}$,
\begin{align*}
    \tildereg^T \le 4 + \frac{(d+1)\log T}{\eta} + 5\eta\sum_{t=1}^T \dnt{ \tilde{\vec\ut}\^t - \tilde{\vec\ut}\^{t-1} }^2 - \frac{1}{27\eta}\sum_{t=1}^{T-1} \mleft\|\vstack{\lambda_\star\^{t+1}}{\vec{y}_\star\^{t+1}} - \vstack{\lambda\^t_\star}{\vec{y}\^t_\star}\mright\|_{(\lambda_\star\^t, \vec y_\star\^t)} \\
    + 2 \sum_{t=1}^T \mleft\|\vstack{\lambda\^t}{\vec y\^t} - \vstack{\lambda\^t_\star}{\vec y\^t_\star}\mright\|_{(\lambda_\star\^t, \vec y_\star\^t)},
\end{align*}
where
    \[
    \vstack{\lambda\^t_\star}{\vec y\^t_\star}\defeq
    \argmax_{(\lambda, \vec y)\in\tilde\cX}\mleft\{ {\eta} \mleft\langle \tilde{\Ut}\^t + \tilde{\vec\ut}\^{t-1}, \vstack{\lambda}{\vec y} \mright\rangle + \log \lambda + \sum_{\ind=1}^d \log \vec{y}[\ind]\mright\}.
    \]
\end{proposition}

\begin{proof}
Fix any $(\lambda^*, \vec{y}^*) \in \tilde \cX$. Then, %the regret $\tildereg^T$ with respect to $(\lambda^*, \vec{y}^*)$ can be expressed as
\begin{align*}
\sum_{t=1}^T \mleft\langle \tilde{\vec\ut}\^t, \vstack{\lambda^*}{\vec{y}^*} - \vstack{\lambda\^t}{\vec y\^t}\mright\rangle &= \sum_{t=1}^T \mleft\langle \tilde{\vec\ut}\^t, \vstack{\lambda^*}{\vec{y}^*} - \vstack{\lambda_\star\^t}{\vec y_\star\^t}\mright\rangle + \sum_{t=1}^T \mleft\langle \tilde{\vec\ut}\^t, \vstack{\lambda_\star^{(t)}}{\vec{y}_\star^{(t)}} - \vstack{\lambda\^t}{\vec y\^t}\mright\rangle \\
&\leq \sum_{t=1}^T \mleft\langle \tilde{\vec\ut}\^t, \vstack{\lambda^*}{\vec{y}^*} - \vstack{\lambda_\star\^t}{\vec y_\star\^t}\mright\rangle + 2 \sum_{t=1}^T \mleft\|\vstack{\lambda\^t}{\vec y\^t} - \vstack{\lambda\^t_\star}{\vec y\^t_\star}\mright\|_{(\lambda_\star\^t, \vec y_\star\^t)},
\end{align*}
where the last inequality uses H\"older's inequality along with the fact that $\|\tilu^{(t)}\|_{*, (\lambda^{(t)}, \vec{y}^{(t)})} \leq 2$, which in turn follows since $\|\vec{u}^{(t)}\|_\infty \|\cX\|_1 \leq 1$ (by assumption). Finally, the proof follows as an immediate consequence of \Cref{thm:RVU}.
\end{proof}

We next proceed with the extension of \Cref{lemma:gamma}.

\begin{lemma}[Extension of \Cref{lemma:gamma}]
    \label{lemma:approx-gamma}
    Suppose that $\epsilon^{(t)} \leq \frac{1}{8}$, for any $t \in \range{T}$. Then, for any time $t \in \range{T-1}$ and learning rate $\eta \leq \frac{1}{256}$,
    \[
        \|\vec x\^{t+1} - \vec x\^t\|_1 \le 8 \|\cX\|_1 \mleft\|\vstack{\lambda_\star\^{t+1}}{\vec{y}_\star\^{t+1}} - \vstack{\lambda\^t_\star}{\vec{y}\^t_\star}\mright\|_{(\lambda_\star\^t, \vec y_\star\^t)} + 16 \|\cX\|_1 \epsilon^{(t+1)} + 8 \|\cX\|_1 \epsilon^{(t)},
    \]
    where $\vec{x}^{(t)} \defeq \vec{y}^{(t)}/ \lambda^{(t)}$.
\end{lemma}

\begin{proof}
First, by the triangle inequality,
\begin{align*}
    \mleft\|\vstack{\lambda_\star\^{t+1}}{\vec{y}_\star\^{t+1}} - \vstack{\lambda\^t_\star}{\vec{y}\^t_\star}\mright\|_{(\lambda_\star\^t, \vec y_\star\^t)} &\geq  \mleft\|\vstack{\lambda\^{t+1}}{\vec{y}\^{t+1}} - \vstack{\lambda\^t}{\vec{y}\^t}\mright\|_{(\lambda_\star\^t, \vec y_\star\^t)} \\
    &- \mleft\|\vstack{\lambda_\star\^{t+1}}{\vec{y}_\star\^{t+1}} - \vstack{\lambda\^{t+1}}{\vec{y}\^{t+1}}\mright\|_{(\lambda_\star\^t, \vec y_\star\^t)} - \mleft\|\vstack{\lambda_\star\^{t}}{\vec{y}_\star\^{t}} - \vstack{\lambda\^{t}}{\vec{y}\^{t}}\mright\|_{(\lambda_\star\^t, \vec y_\star\^t)}.
\end{align*}
Now given that $\eta \leq \frac{1}{50}$, it follows from \Cref{cor:ftrl stability} that
\begin{equation*}
    \mleft\|\vstack{\lambda_\star\^{t+1}}{\vec{y}_\star\^{t+1}} - \vstack{\lambda\^t_\star}{\vec{y}\^t_\star}\mright\|_{(\lambda_\star\^t, \vec y_\star\^t)} \leq \frac{1}{2},
\end{equation*}
which in turn---combined with \Cref{lem:multi_stab}---implies that
\begin{equation*}
    \mleft\|\vstack{\lambda_\star\^{t+1}}{\vec{y}_\star\^{t+1}} - \vstack{\lambda\^{t+1}}{\vec{y}\^{t+1}}\mright\|_{(\lambda_\star\^t, \vec y_\star\^t)} \leq 2 \mleft\|\vstack{\lambda_\star\^{t+1}}{\vec{y}_\star\^{t+1}} - \vstack{\lambda\^{t+1}}{\vec{y}\^{t+1}}\mright\|_{(\lambda_\star\^{t+1}, \vec y_\star\^{t+1})}.
\end{equation*}
Similarly, since $\epsilon^{(t)} \leq \frac{1}{8}$, it follows that
\begin{equation*}
    \mleft\|\vstack{\lambda\^{t+1}}{\vec{y}\^{t+1}} - \vstack{\lambda\^t}{\vec{y}\^t}\mright\|_{(\lambda_\star\^t, \vec y_\star\^t)} \geq \frac{1}{2} \mleft\|\vstack{\lambda\^{t+1}}{\vec{y}\^{t+1}} - \vstack{\lambda\^t}{\vec{y}\^t}\mright\|_{(\lambda\^t, \vec y\^t)}.
\end{equation*}
As a result,
\begin{equation}
    \label{eq:penl}
    \mleft\|\vstack{\lambda_\star\^{t+1}}{\vec{y}_\star\^{t+1}} - \vstack{\lambda\^t_\star}{\vec{y}\^t_\star}\mright\|_{(\lambda_\star\^t, \vec y_\star\^t)} \geq \frac{1}{2} \mleft\|\vstack{\lambda\^{t+1}}{\vec{y}\^{t+1}} - \vstack{\lambda\^t}{\vec{y}\^t}\mright\|_{(\lambda\^t, \vec y\^t)} - 2 \epsilon^{(t+1)} - \epsilon^{(t)}.
\end{equation}
Next, we will prove that
\begin{equation}
    \label{eq:mul-claim}
    \max \left\{ \left| \frac{\lambda^{(t+1)}}{\lambda^{(t)}} - 1 \right|, \max_{\ind \in \range{d}} \left| \frac{\vec{y}^{(t+1)}[\ind]}{\vec{y}^{(t)}[\ind]} - 1 \right| \right\} \leq \frac{1}{2}.
\end{equation}
Indeed, since $\epsilon^{(t)}, \epsilon^{(t+1)} \leq \frac{1}{8}$, it holds that
\begin{equation*}
    \left| 1 - \frac{\lambda^{(t)}}{\lambda_\star^{(t)}} \right| \leq \frac{1}{8} \implies \frac{7}{8} \lambda_\star^{(t)} \leq \lambda^{(t)} \leq \frac{9}{8} \lambda_\star^{(t)},
\end{equation*}
and
\begin{equation*}
    \left| 1 - \frac{\lambda^{(t+1)}}{\lambda_\star^{(t+1)}} \right| \leq \frac{1}{8} \implies \frac{7}{8} \lambda_\star^{(t+1)} \leq \lambda^{(t+1)} \leq \frac{9}{8} \lambda_\star^{(t+1)}.
\end{equation*}
Furthermore, for $\eta \leq \frac{1}{256}$,
\begin{equation*}
    \left| 1 - \frac{\lambda_\star^{(t+1)}}{\lambda_\star^{(t)}} \right| \leq \frac{1}{10} \implies \frac{9}{10} \lambda_\star^{(t)} \leq \lambda_\star^{(t+1)} \leq \frac{11}{10} \lambda_\star^{(t)},
\end{equation*}
by \Cref{cor:ftrl stability} and \Cref{lem:multi_stab}. Thus,
\begin{equation*}
    \frac{2}{3} \lambda^{(t+1)} \leq \frac{7}{8} \frac{10}{11} \frac{8}{9} \lambda^{(t+1)} \leq  \lambda^{(t)} \leq \frac{9}{8} \frac{10}{9} \frac{8}{7} \lambda^{(t+1)} \leq \frac{3}{2} \lambda^{(t+1)},
\end{equation*}
in turn implying that 
\begin{equation*}
    \left| 1 - \frac{\lambda^{(t+1)}}{\lambda^{(t)}} \right| \leq \frac{1}{2}.
\end{equation*}
Similarly, we conclude that for any $\ind \in \range{d}$,
\begin{equation*}
    \left| 1 - \frac{\vec{y}^{(t+1)}[\ind]}{\vec{y}^{(t)}[\ind]} \right| \leq \frac{1}{2},
\end{equation*}
confirming \eqref{eq:mul-claim}. Hence, following the proof of \Cref{lem:ratio close}, we derive that
\begin{equation*}
    \mleft\|\vstack{\lambda\^{t+1}}{\vec{y}\^{t+1}} - \vstack{\lambda\^t}{\vec{y}\^t}\mright\|_{(\lambda\^t, \vec y\^t)} \geq \frac{1}{4\|\cX\|_1} \left\| \frac{\vec{y}^{(t+1)}}{\lambda^{(t+1)}} - \frac{\vec{y}^{(t)}}{\lambda^{(t)}} \right\|_1 = \frac{1}{4\|\cX\|_1} \|\vec{x}^{(t+1)} - \vec{x}^{(t)} \|_1.
\end{equation*}
Combining this bound with \eqref{eq:penl} concludes the proof.
\end{proof}

We also state the following immediate implication of \Cref{lemma:approx-gamma}.

\begin{corollary}
    Suppose that $\epsilon^{(t)} \leq \frac{1}{8}$, for any $t \in \range{T}$. Then, for any $t \in \range{T-1}$ and learning rate $\eta \leq \frac{1}{256}$,
    \begin{equation*}
        \|\vec{x}^{(t+1)} - \vec{x}^{(t)} \|_1^2 \leq 192 \|\cX\|_1^2 \mleft\|\vstack{\lambda_\star\^{t+1}}{\vec{y}_\star\^{t+1}} - \vstack{\lambda\^t_\star}{\vec{y}\^t_\star}\mright\|_{(\lambda_\star\^t, \vec y_\star\^t)}^2 + 768 \|\cX\|_1^2 (\epsilon^{(t+1)})^2 + 192 \|\cX\|_1^2 (\epsilon^{(t)})^2,
    \end{equation*}
    where $\vec{x}^{(t)} \defeq \vec{y}^{(t)}/ \lambda^{(t)}$.
\end{corollary}
As a result, combining this bound with \Cref{prop:approx-rvu} extends \Cref{cor:rvu} with an error term proportional to $\sum_{t=1}^T \epsilon^{(t)}$. Finally, the rest of the extension is identical to our proof of \Cref{theorem:main-detailed}.

\section{Implementation via Proximal Oracles}\label{app:proximal}

In this section we provide the omitted proofs from \Cref{sec:implementation} regarding the implementation of \algoshort using proximal oracles (recall \Cref{eq:quad_approx}).

\subsection{The Proximal Newton Method}
\label{sec:newton}

In this subsection we describe the proximal Newton algorithm of~\citet{Dinh15:Composite}, leading to \Cref{theorem:prox-Newton} we presented in \Cref{sec:implementation}. More precisely, \citet{Dinh15:Composite} studied the following composite minimization problem:
\begin{equation}
	\label{eq:comp}
	\min_{\tilx \in \R^{d+1}} \left\{ F(\tilx) \defeq f(\tilx) + g(\tilx) \right\},
\end{equation}
where $f$ is a (standard) self-concordant and convex function, and $g : \R^{d+1} \to \R \cup \{+\infty \}$ is a proper, closed and convex function. In our setting, we will let $g$ be defined as
\begin{equation*}
	g(\tilx) \defeq
	\begin{cases}
		0       & \text{if } \tilx \in \Tilde{\cX}, \\
		+\infty & \text{otherwise}.
	\end{cases}
\end{equation*}

Further, for a given time $t \in \N$, we let
\begin{equation*}
	f : \tilde{\vec{x}} \mapsto - \eta \mleft\langle \tilde{\Ut}\^t + \tilde{\vec\ut}\^{t-1}, \tilde{\vec{x}} \mright\rangle - \sum_{\ind=1}^{d+1} \log \tilde{\vec{x}}[\ind].
\end{equation*}
Before we describe the proximal Newton method, let us define $\tils_k$ as follows.

\begin{equation}
	\label{eq:s}
	\tils_k \defeq \argmin_{\tilx \in \tilde \cX} \left\{ f(\tilx_k) + (\nabla f(\tilx_k))^\top (\tilx - \tilx_k) + \frac{1}{2} (\tilx - \tilx_k)^\top \nabla^2 f(\tilx_k) (\tilx - \tilx_k) \right\},
\end{equation}
for some $\tilx_k \in \R^{d+1}_{> 0}$. We point out that the optimization problem \eqref{eq:s} can be trivially solved when we have access to a (local) proximal oracle---given in \Cref{eq:quad_approx}.

In this context, the proximal Newton method of~\citet{Dinh15:Composite} is given in~\Cref{algo:newton}. Their algorithm proceeds in two phases. In the first phase we perform \emph{damped steps} of proximal Newton until we reach the region of quadratic convergence. Afterwards, we perform \emph{full steps} of proximal Newton until the desired precision $\epsilon > 0$ has been reached. Below we summarize the main guarantee regarding \Cref{algo:newton}, namely~\citep[Theorem 9]{Dinh15:Composite}.

\begin{theorem}[~\citep{Dinh15:Composite}]
	\label{theorem:composite}
	\Cref{algo:newton} returns $\tilx_K \in \R^{d+1}_{> 0}$ such that $\|\tilx_K - \tilx^*\|_{\tilx^*} \leq 2 \epsilon$ after at most
	\begin{equation*}
		K = \left\lfloor \frac{f(\tilx_0) - f(\tilx^*)}{0.017} \right\rfloor + \left\lfloor 1.5 \ln \ln \left( \frac{0.28}{\epsilon} \right) \right\rfloor + 2
	\end{equation*}
	iterations, for any $\epsilon > 0$, where $\tilx^* = \argmin_{\tilx} F(\tilx)$, for the composite function $F$ defined in \eqref{eq:comp}.
\end{theorem}

To establish \Cref{theorem:prox-Newton} from this guarantee, it suffices to initialize \Cref{algo:newton} at every iteration $t \geq 2$ with $\tilx_0 \defeq \tilx^{(t-1)} = (\lambda^{(t-1)}, \vec{y}^{(t-1)})$. Then, as long as $\epsilon^{(t-1)}$ is sufficiently small, the number of iterations predicted by \Cref{theorem:composite} will be bounded by $O(\log \log (1/\epsilon))$, in turn establishing \Cref{theorem:prox-Newton}.

\begin{algorithm}[htp]
	\caption{Proximal Newton~\citep{Dinh15:Composite}}
	\label{algo:newton}
	\DontPrintSemicolon
	\KwData{Initial point $\tilx_{0}$ \\
		Precision $\epsilon > 0$ \\
		Constant $\sigma \defeq 0.2$
	}\vspace{2mm}
	\For{$k = 1, \dots, K$}{
	Obtain the proximal Newton direction $\tild_k \gets \tils_k - \tilx_k $, where $\tils_k$ is defined in \eqref{eq:s} \;
	Set $\lambda_k \gets \| \tild_k \|_{\tilx_k}$\;
	\If{$\lambda_k > 0.2$}{
		$\tilx_{k+1} \gets \tilx_k + \alpha_k \tild_k$, where $\alpha_k \defeq (1 + \lambda_k)^{-1}$\Comment*{\color{commentcolor}Damped Step]\!\!\!\!\!}
	}
	\ElseIf{$\lambda_k > \epsilon$}{
		$\tilx_{k+1} \gets \tilx_k + \tild_k$ \Comment*{\color{commentcolor}Full Step]\!\!\!\!\!}
	}
	\Else{\textbf{return} $ \tilx_k$}
	}
\end{algorithm}

\subsection{Proximal Oracle for Normal-Form and Extensive-Form Games}

In order to show that the proximal oracle of \Cref{sec:implementation} can be implemented efficiently for probability simplexes (\emph{i.e.}, the strategy sets of normal-form games) and sequence-form strategy spaces (\emph{i.e.}, the strategy sets of extensive-form games), we will prove a slightly stronger result concerning \emph{treeplex} sets, of which sequence-form strategy spaces are instances.

\begin{definition}
	A set $Q \subseteq [0,+\infty)^d$, $d \ge 1$, is \emph{treeplex} if it is:
	\begin{enumerate}
		\item a simplex $Q = \Delta^d$;
		\item a Cartesian product of treeplex sets $Q_1 \times \dots \times Q_K$; or
		\item (Branching operation) a set of the form
		      \[
			      \triangle(Q_1,\dots,Q_K) \defeq \{(\vec{x}, \vec{x}[1]\vec{q}_1, \dots, \vec{x}[K]\vec{q}_K): \vec{x} \in \Delta^K, \vec{q}_k \in Q_k ~~\forall k \in \range K\},
		      \]
		      where $Q_1,\dots,Q_K$ are treeplex.
	\end{enumerate}
\end{definition}

We will show that any treeplex $Q$ is such that $[0,1]Q$ admits an efficient (positive-definite) quadratic optimization oracle. This is sufficient, since it is well-known that every sequence-form strategy space $\cX$ is treeplex (e.g., \citet{Hoda10:Smoothing}) and therefore, by definition, so is the set $\{(1, \vec x) : \vec x \in \cX\}$.

Introduce the \emph{value function}
\NewDocumentCommand{\val}{O{t}mmm}{V_{#2}(#1; #3, #4)}
\NewDocumentCommand{\lam}{O{t}mmm}{\lambda_{#2}(#1; #3, #4)}
\NewDocumentCommand{\bk}{O{t}mmm}{B_{#2}(#1; #3, #4)}
\begin{align*}
	\val{Q}{\vec g}{\vec w} \defeq \min_{\vec x \in t Q} \left\{ - \vec g^\top \vx + \frac{1}{2}\sum_{\ind=1}^d \mleft(\frac{\vec{x}[\ind]}{\vec{w}[\ind]}\mright)^2 \right\} \qquad\qquad (t \geq 0, \vec w > \vec 0)
	\numberthis{eq:prox}
\end{align*}
(note the rescaling by $t$ in the domain of the minimization).
We will be interested in the derivative of $\val{Q}{\vec g}{\vec w}$, which we will denote as\footnote{For $t=0$ we define $\lam{Q}{\vec g}{\vec w}$ in the usual way as
	\[
		\lam[0]{Q}{\vec g}{\vec w} = \lim_{t\to 0^+} \frac{\val{Q}{\vec g}{\vec w} - \val[0]{Q}{\vec g}{\vec w}}{t} = \lim_{t\to 0^+} \frac{\val{Q}{\vec g}{\vec w}}{t}.
	\]
}
\[
	\lam{Q}{\vec g}{\vec w} \defeq \frac{d}{dt}\val{Q}{\vec g}{\vec w}.
\]

% as well as the inverse of $\lambda_{Q,g}$ which we denote by $\theta_{Q,g}$.
% The notation $\lambda_{Q,g}$ for the derivative function evokes the fact that we will use $\lambda\in \R_{\ge0}$ to be a certain Lagrange multiplier which will turn out to be exactly the derivative. 

\newcommand{\smpl}{SMPL\xspace}
\paragraph{Preliminaries on Strictly Monotonic Piecewise-Linear (\smpl) Functions}
\begin{definition}[\smpl function and standard representation]
	Given an interval $I\subseteq \R$ and a function $f : I \to \bbR$, we say that $f$ is \smpl if it is strictly monotonically increasing and piecewise-linear on $I$.
\end{definition}
\begin{definition}[Quasi-\smpl function]
	A quasi-\smpl function is a function $f: \R \to [0,+\infty)$ of the form $f(x) = [g(x)]^+$ where $g(x): \R \to \R$ is \smpl and $[\,\cdot\,]^+\defeq\max\{0, \,\cdot\,\}$.
\end{definition}
\begin{definition}
	Given a \smpl or quasi-\smpl function $f$, a \emph{standard representation} for it is an expression of the form
	\[
		f(x) = \zeta + \alpha_0 x + \sum_{s=1}^S \alpha_s [x - \beta_s]^+,
	\]
	valid for all $x$ in the domain of $f$, where $S \in \N \cup \{0\}$ and $\beta_1 < \cdots < \beta_S$. The size of the standard representation is defined as the natural number $S$.
\end{definition}

We now mention four basic results about \smpl and quasi-\smpl functions. The proofs are elementary and omitted.

\begin{lemma}\label{lem:smpl mul}
	Let $f: I \to \R$ be \smpl, and consider a standard representation of $f$ of size $S$. Then, for any $\zeta\in\R$ and $\alpha \ge 0$, a standard representation for the \smpl function $I \ni x\mapsto \zeta + \alpha f(x)$ can be computed in $O(S+1)$ time.
\end{lemma}
\begin{lemma}\label{lem:quasi smpl sum}
	The sum $f_1+\dots+f_n$ of $n$ \smpl (resp., quasi-\smpl) functions $f_i : I\to\R$ is a \smpl (resp., quasi-\smpl) function $I\to\R$. Furthermore, if each $f_i$ admits a standard representation of size $S_i$, then a standard representation of size at most $S_1 + \dots + S_n$ for their sum can be computed in $O((S_1 + \dots + S_n + 1) \log n)$ time.
\end{lemma}
\begin{lemma}\label{lem:smpl thresh}
	Let $f: \R \to \R$ be \smpl, and consider a standard representation of $f$ of size $S$. Then, for any $\beta \in \bbR$, a standard representation of size at most $S$ for the quasi-\smpl function $I \ni x\mapsto [f(x) - \beta]^+$ can be computed in $O(S+1)$ time.
\end{lemma}
\begin{lemma}\label{lem:smpl inv}
	The inverse $f^{-1}: \mathrm{range}(f) \to \R$ of a \smpl function $f: I \to \R$ is \smpl. Furthermore, if $f$ admits a standard representation of size $S$, then a standard representation for $f^{-1}$ of size at most $S$ can be computed in $O(S + 1)$ time.
\end{lemma}
\begin{lemma}\label{lem:quasi smpl inv}
	Let $f: \R \to [0, +\infty)$ be quasi-\smpl.
	The restricted inverse $f^{-1}: (0, +\infty) \to \R$  of $f$ is \smpl, where we restrict the domain to $(0,+\infty)$ because $f^{-1}(0)$ may be multivalued. Furthermore, if $f$ admits a standard representation of size $S$, then a standard representation of size at most $S$ for $f^{-1}$ can be computed in $O(S + 1)$ time.
\end{lemma}
\begin{proof}
	We have $f(x) = [g(x)]^+$ where $g$ is SMPL. It follows that the function $\bar g: I \rightarrow \R$ defined as $\bar g(x) = g(x)$ for the interval $I= \{x : g(x) > 0\}$ is SMPL as well.
	For any $x$ such that $f(x) > 0$ we have $x\in I$, and thus $f^{-1} = g^{-1}$, and it follows from \cref{lem:smpl inv} that $f^{-1}$ is SMPL.
\end{proof}

\begin{lemma}\label{lem:smpl fp}
	Let $f: [0,+\infty) \to \R$ be a \smpl function, and consider the function $g$ that maps $y$ to the unique solution to the equation $x = [y - f(x)]^+$. Then, $g$ is quasi-\smpl and satisfies $g(y) = [(x + f)^{-1}(y)]^+$, where $(x + f)^{-1}$ denotes the inverse of the \smpl function $x \mapsto x + f([x]^+)$.
\end{lemma}
\begin{proof}
	For any $y \in \R$, the function $h_y : x \mapsto x - [y - f(x)]^+$ is clearly \smpl on $[0, +\infty)$. Furthermore, $h_y(0) \le 0$ and $h_y(+\infty) = +\infty$, implying that $h_y(x) = 0$ has a unique solution. We now show that $g(y) = [(x+f)^{-1}(y)]^+$ is that solution, that is, it satisfies $g(y) = [y - f(g(y))]^+$ for all $y \in \R$. Fix any $y \in \R$ and let
	\[
		\bar g \defeq (x + f)^{-1} (y)
		\quad\iff\quad
		\bar g + f([\bar g]^+) = y
		\quad\iff\quad
		\bar g = y - f([\bar g]^+)\numberthis{eq:barg}
	\]
	There are two cases:
	\begin{itemize}
		\item If $\bar g \ge 0$, then $g(y) = [\bar g]^+ = \bar g$, and so we have
		      \[
			      g(y) = [\bar g]^+ = [y - f([\bar g]^+)]^+ = [y - f(g(y))]^+,
		      \]
		      as we wanted to show.
		\item Otherwise, $\bar g < 0$ and $g(y) = 0$. From~\eqref{eq:barg}, the condition $\bar g < 0$ implies $y < f([\bar g]^+) = f(0)$. So, it is indeed the case that
		      \[
			      0 = g(y) = [y - f(0)]^+ = [y - f(g(0))]^+,
		      \]
		      as we wanted to show.
	\end{itemize}
	Finally, we note that the function $(x+f)^{-1} : \R \to \R$ is \smpl due to \cref{lem:smpl inv}, implying that $g(y)$ is quasi-\smpl.
\end{proof}

\paragraph{Central result}
The following result is central in our analysis.
\begin{lemma}\label{lem:central}
	For any treeplex $Q \subseteq \bbR^d$, gradient $\vec g \in \bbR^d$, and center $\vec w \in \bbR^d_{>0}$, the function $t\mapsto \lam{Q}{\vec g}{\vec w}$ is \smpl, and a standard representation of it of size $d$ can be computed in polynomial time in $d$.
\end{lemma}
\begin{proof}
	We will prove the result by structural induction on $Q$.
	\begin{itemize}[left=3mm]
		\item First, we consider the case where $Q$ is a Cartesian product,
		      \[
			      Q = Q_1 \times \dots \times Q_K.
		      \]
		      In that case, the value function decomposes as follows
		      \begin{align*}
			      \val{Q}{\vec g}{\vec w} = \sum_{k=1}^K \min_{\vec x_k\in tQ_k} \mleft\{ -\vec g_k^\top \vec x_k + \frac{1}{2}\sum_{\ind=1}^{d_k} \mleft(\frac{\vec x_k[\ind]}{\vec w_k[\ind]}\mright)^2 \mright\} = \sum_{k=1}^K \val{Q_k}{\vec g_k}{\vec w_k}.
		      \end{align*}
		      By linearity of derivatives, we have
		      \begin{align*}
			      \lam{Q}{\vec g}{\vec w} = \sum_{k=1}^K \lam{Q_k}{\vec g_k}{\vec w_k}.
		      \end{align*}
		      From \cref{lem:quasi smpl sum}, we conclude that $\lam{Q}{\vec g}{\vec w}$ is a \smpl function with domain $[0,+\infty)$ which admits a standard representation of size at most $d = d_1 + \dots +d_K$ computable in time $O(d \log K)$ starting from the standard representation of each of the $\lam{Q_k}{\vec g_k}{\vec w_k}$.

		\item Second, consider the case where $Q$ is a simplex or the result of a branching operation
		      \[
			      \triangle(Q_1,\dots,Q_K) = \{(\vec{x}, \vec{x}[1]\vec{q}_1, \dots, \vec{x}[K]\vec{q}_K): \vec{x} \in \Delta^K, \vec{q}_k \in Q_k ~~\forall k \in \range K\},
		      \]
		      where $Q_k \in \bbR^{d_k}$. With a slighty abuse of notation, we will treat the two cases together, considering the $K$-simplex $\Delta^K$ as a branching operation over empty sets $Q_k = \emptyset$.

		      In this case, we can write
		      \newcommand{\bul}{\bullet}
		      \[
			      \vec g &= (\vec g_\bul[1], \dots, \vec g_\bul[K], \vec g_1 \in \bbR^{d_1}, \cdots, \vec g_K \in \bbR^{d_K}),\text{ and}\\
			      \vec w &= (\vec w_\bul[1], \dots, \vec w_\bul[K], \vec w_1 \in \bbR_{>0}^{d_1}, \cdots, \vec w_K \in \bbR_{>0}^{d_K}).
		      \]
		      The value function then decomposes recursively as
		      \[
			      \val{Q}{\vec g}{\vec w} &= \min_{\vec x_\bul\in t\Delta^K} \mleft\{\mleft( -\sum_{\ind=1}^K \vec g_\bul[\ind] \vec x_\bul[\ind] + \frac{1}{2}\sum_{\ind=1}^K \mleft(\frac{\vec x_\bul[\ind]}{\vec w_\bul[\ind]}\mright)^2\mright)\mright.\\
			      &\hspace{3cm} \mleft.+ \sum_{k=1}^K \min_{\vec x_k \in \vec x_\bul[k] Q_k} \mleft\{ - \vec g_k^\top \vec x_k + \sum_{\ind=1}^{d_k} \mleft(\frac{\vec x_k[\ind]}{\vec w_k[\ind]}\mright)^2 \mright\}\mright\}\\
			      &= \min_{\vec x_\bul\in t\Delta^K} \mleft\{\mleft( -\sum_{\ind=1}^K \vec g_\bul[\ind] \vec x_\bul[\ind] + \frac{1}{2}\sum_{\ind=1}^K \mleft(\frac{\vec x_\bul[\ind]}{\vec w_\bul[\ind]}\mright)^2\mright) + \val[\vec x_\bul[k]]{Q_k}{\vec g_k}{\vec w_k}\mright\}.\numberthis{eq:vf branch}
		      \]

		      Suppose that for each $k\in \range{K}$, $\lam{Q_k}{\vec g_k}{\vec w_k}$ is piecewise linear and monotonically increasing in $t$.
		      Now we consider the KKT conditions for $\vec x_\bul$ in \cref{eq:vf branch}:
		      \[
			      -\vec g_\bul[k] + \frac{\vx_\bul[k]}{\vec w_\bul[k]^2} + \lam[\vx_\bul[k]]{Q_k}{\vg_k}{\vec w_k} = \lambda_\bul + \vec\mu[k] &\qquad \forall k \in \range K && \text{(Stationarity)} \\
			      \vx_\bul \in t\cdot \Delta^K &&& \text{(Primal feasibility)} \\
			      \lambda_\bul \in \R, \vec\mu \in \R_{\ge0}^d &&& \text{(Dual feasibility)}\\
			      \vec \mu[k] \vx_\bul[k] = 0 & \qquad\forall k \in \range K&& \text{(Compl. slackness)}
		      \]
		      Solving for $\vx_\bul[k]$ in the stationarity condition, and using the conditions $\vx_\bul[k]\vec\mu[k] = 0$ and $\vec\mu[k] \geq 0$, it follows that for all $k\in\range K$
		      \[
			      \vx_\bul[k] &= \vec w_\bul[k]^2\Big(\lambda_\bul + \vec\mu[k] + \vg_\bul[k] - \lam[\vx_\bul[k]]{Q_k}{\vg_k}{\vec w_k}\Big)\\
			      &= \vec w_\bul[k]^2\Big[\lambda_\bul + \vg_\bul[k]  - \lam[\vx_\bul[k]]{Q_k}{\vg_k}{\vec w_k}\Big]^+.\numberthis{eq:stationarity no mu}
		      \]
		      %where $[\,\cdot\,]^+\defeq\max\{0, \,\cdot\,\}$.

		      %% Step 1: solve for x[k] given lambda. Show uniqueness and breakpoints.
		      \paragraph{Strict monotonicity and piecewise-linearity of $\vec x_\bul[k]$ as a function of $\lambda_\bul$.}

		      Given the preliminaries on \smpl functions, it is now immediate to see that $\vx_\bul[k]$ is unique as a function of $\lambda_\bul$. Indeed, note that \eqref{eq:stationarity no mu} can be rewritten as
		      \[
			      \vx_\bul[k] = \Big[(\vec w_\bul[k]^2)\lambda_\bul - \vec w_\bul[k]^2\mleft(-\vg_\bul[k] + \lam[\vx_\bul[k]]{Q_k}{\vg_k}{\vec w_k}\mright)\Big]^+,
		      \]
		      which is a fixed-point problem of the form studied in \cref{lem:smpl fp} for $y = (\vec w_\bul[k]^2)\lambda_\bul$ and function $f_k$ defined as
		      \[
			      f_k(\vx_\bul[k]) = \vec w_\bul[k]^2\mleft(-\vg_\bul[k] + \lam[\vx_\bul[k]]{Q_k}{\vg_k}{\vec w_k}\mright),
		      \]
		      which is clearly \smpl by inductive hypothesis. Hence, the unique solution to the previous fixed-point equation is given by the quasi-\smpl function
		      \[
			      g_k: \lambda_\bul \mapsto \frac{1}{\vec w_\bul[k]^2}\mleft[(\vx_\bul[k]+f_k)^{-1}(\lambda_\bul)\mright]^+,
		      \]
		      a standard representation of which can be computed in time $O(d + 1)$ by combining the results of \cref{lem:smpl mul,lem:smpl thresh,lem:smpl inv} given that a standard representation of $\lam{Q_k}{\vg_k}{\vec w_k}$ of size $d$ is available by inductive hypothesis.

		      % 	Thus, in order to satisfy $\vx_\bul \in t\cdot \Delta^K$ we must find $\lambda_\bul$ such that 
		      % 	\[
		      % 		f(\lambda_\bul) \defeq \sum_{i=1}^K w_\bul[k]^2\Big[\lambda_\bul + \vg_\bul[k]  - \lam[\vx_\bul[k]]{Q_k}{\vg_k}{\vec w_k}\Big]^+ = t.
		      % 	\]

		      % 	\todo{not immediate}
		      % 	It is immediate to see that $f_\lambda$ is piecewise-linear and weakly increasing. Let
		      % 	\[
		      % 	    \bar\lambda \defeq \min_{k\in\range K}\Big\{-\vg_\bul[k] + \lam[0]{Q_k}{\vg_k}{\vec w_k}\Big\}.
		      % 	\] 
		      % 	When $\lambda_\bul \le \bar\lambda$, then $f(\lambda_\bul) = 0$, whereas for $\lambda_\bul \ge \bar\lambda$, $f(\lambda_\bul)$ is strictly increasing and grows to $+\infty$ as $\lambda_\bul\to+\infty$. This implies that whenever $t > 0$, there exists a unique value $\lambda_\bul^\star(t)$ such that $f(\lambda_\bul^\star) = t$.

		      \paragraph{Strict monotonicity and piecewise-linearity of $\lambda_\bullet$ as a function of $t$.}

		      At this stage, we know that given any value of the dual variable $\lambda_\bul$, the unique value of the coordinate $\vec x_\bul[k]$ that solves the KKT system can be computed using the quasi-\smpl function $g_k$. In turn, this means that we can remove the primal variables $\vec x_\bul$ from the KKT system, leaving us a system in $\lambda_\bul$ and $t$ only. We now show that the solution $\lambda_\bul^\star$ of that system is a \smpl function of $t\in [0, +\infty)$.

		      Indeed, the value of $\lambda^\star_\bul$ that solves the KKT system has to satisfy the primal feasibility condition
		      \[
			      t = \sum_{k = 1}^K \vx_\bul[k] = \sum_{k=1}^K g_k(\lambda_\bul).
		      \]
		      Fix any $t > 0$. The right-hand side of the equation is a sum of quasi-\smpl functions. Hence, from \cref{lem:quasi smpl sum}, we have that the right-hand side has a standard representation of size at most $K+\sum_{k=1}^K d_k = d$ can be computed in time $O(d \log K)$.
		      Furthermore, from \cref{lem:quasi smpl inv}, we have that the $\lambda^\star_\bul$ that satisfies the equation is unique, and in fact that the mapping $(0,+\infty) \ni t \mapsto \lambda^\star_\bul$ is \smpl with standard representation of size at most $d$.

		      \paragraph{Relating $\lambda_\bul$ and $\lam{Q}{\vg}{\vec w}$.}
		      Since $\lambda_\bul^\star(t)$ is the coefficient on $t$ in the Lagrangian relaxation of \eqref{eq:vf branch}, it is a subgradient of $\val{Q}{\vg}{\vec w}$, and since there is a unique solution, we get that it is the derivative, that is,
		      \[\lambda_\bul^\star(t) = \lam{Q}{\vg}{\vec w}\]
		      for all $t \in (0, +\infty)$. To conclude the proof by induction, we then need to analyze the case $t=0$, which has so far been excluded. When $t=0$, the feasible set $tQ$ is a singleton, and $\val[0]{Q}{\vg}{\vec w} = 0$. Since $\val{Q}{\vg}{\vec w}$ is continuous on $[0,+\infty)$, and since $\lim_{t\to 0^+} \lam{Q}{\vg}{\vec w} = \lim_{t\to0^+} \lambda_\bul^\star(t)$ exists since $\lambda_\bul^\star(t)$ is piecewise-linear, then by the mean value theorem,
		      \[
			      \lam[0]{Q}{\vg}{\vec w} = \lim_{t\to0^+} \lambda_\bul^\star(t),
		      \]
		      that is, the continuous extension of $\lambda_\bul^\star$ must be (right) derivative of $\val{Q}{\vg}{\vec w}$ in $0$. As extending continuously $\lambda_\bul^\star(t)$ clearly does not alter its being \smpl nor its standard representation, we conclude the proof of the inductive case.
	\end{itemize}
	% Finally, we conclude the analysis of the runtime of the algorithm. As we showed, when $Q$ is obtained through a Cartesian product operation $Q = Q_1 +\dots + Q_K$, the time $T$ necessary to compute a standard representation of $\lam{Q}{\vg}{\vec w}$ satisfies the recurrence $T(Q) = O(d \log K) + \sum_{k=1}^K T(Q_k)$. The same recursive relationship applies for the branching operation. 
\end{proof}

% Secondly, consider the case where $Q$ is a branching operation. 

% The first step will be to show that the derivative function $\lambda_{Q,\vg}(t)$ is piecewise linear with a small number of pieces. Here we will count an equivalent number, which is the number of \emph{breakpoints} in $\lambda_{Q,\vg}$, which is the set of points where the slope changes.

\cref{lem:central} also provides a constructive way of computing the argmin of \eqref{eq:prox} in polynomial time for any $t\in[0,+\infty)$. To conclude the construction of the proximal oracle, it is then enough to show how to pick the optimal value of $t \in [0,1]$ that minimizes
\[
	\min_{\vec x \in [0,1] Q} \left\{ - \vec g^\top \vx + \frac{1}{2}\sum_{\ind=1}^d \mleft(\frac{\vec{x}[\ind]}{\vec{w}[\ind]}\mright)^2 \right\} = \min_{t\in [0,1]} \val{Q}{\vec g}{\vec w}.
\]
That is easy starting from the derivative $\lam{Q}{\vec g}{\vec w}$, which is a \smpl function by \cref{lem:central}. Indeed, if $\lam[0]{Q}{\vec g}{\vec w} \ge 0$, then by monotonicity of the derivative we know that the optimal value of $t$ is $t=0$. Else, if $\lam[1]{Q}{\vec g}{\vec w} \le 0$, again by monotonicity we know that the optimal value of $t$ is $t=1$. Else, there exists a unique value of $t\in (0,1)$ at which the derivative of the objective is $0$, and such a value can be computed exactly using \cref{lem:smpl inv}.
\section{Experimental Results}
\label{section:experiments}

In this section we provide preliminary experimental results in order to verify our theoretical findings, and in particular the per-player regret bound established in \Cref{theorem:main-detailed}. More specifically, we investigate the behavior of our learning dynamics (\algoshort) in four standard extensive-form games used in the literature: $2$-player and $3$-player \emph{Kuhn poker}~\citep{Kuhn50:Extensive}; $2$-player \emph{Goofspiel}~\citep{Ross71:Goofspiel};\footnote{We consider instances of Goofspiel with $r = 3$ cards and \emph{limited information}---the actions of the other player are only observed at the end of the game. Also, we note that the tie-breaking mechanism makes the game general-sum.} and the baseline version of ($2$-player) \emph{Sheriff}~\citep{Farina19:Correlation}. From those games, only $2$-player Kuhn poker is a zero-sum game. Our findings are summarized in \Cref{fig:kuhn}.

\begin{figure}[!ht]
    \centering
    \def\sc{.76}
    \includegraphics[scale=\sc]{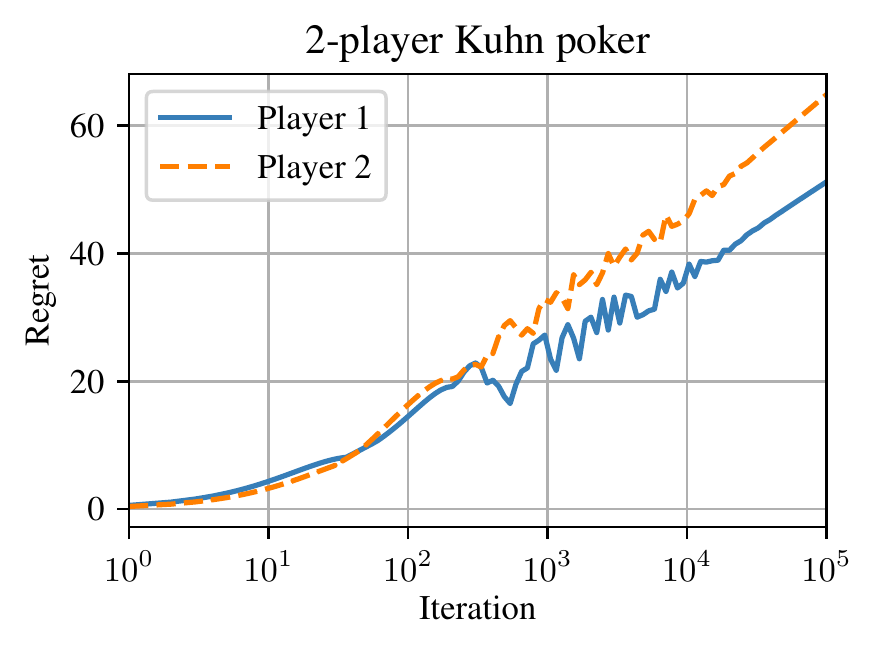}\hfill%
    \includegraphics[scale=\sc]{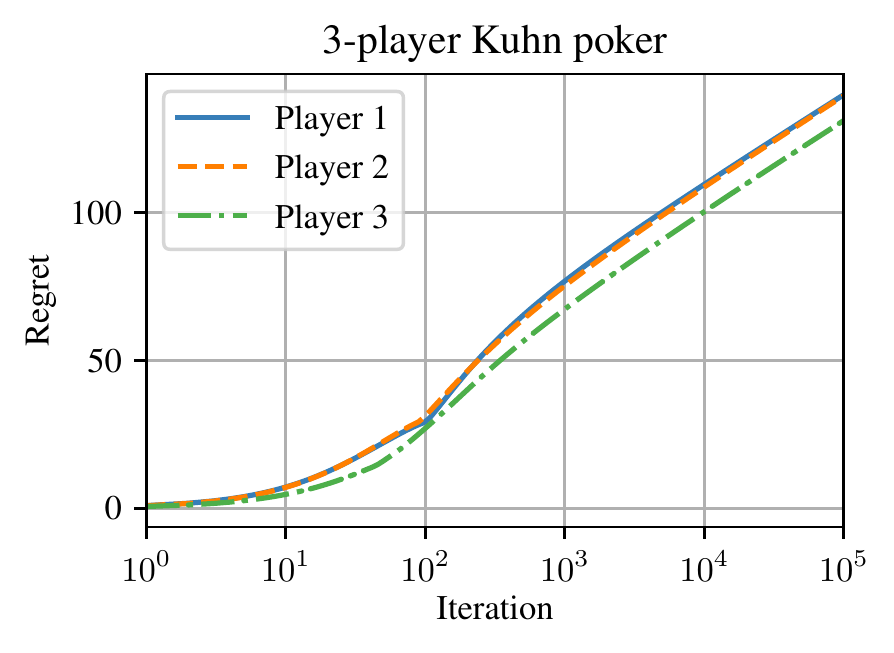}\\
    \includegraphics[scale=\sc]{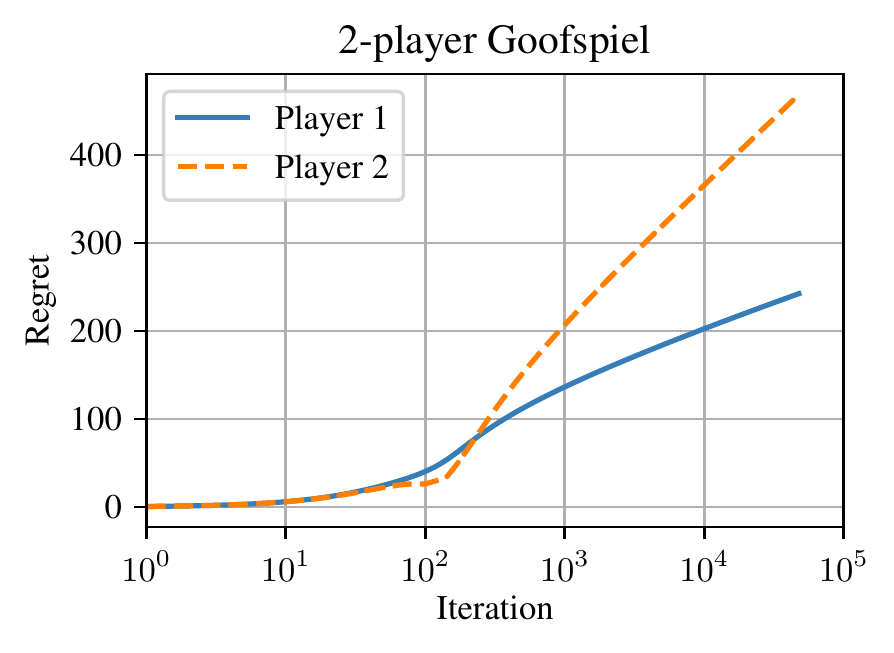}\hfill%
    \includegraphics[scale=\sc]{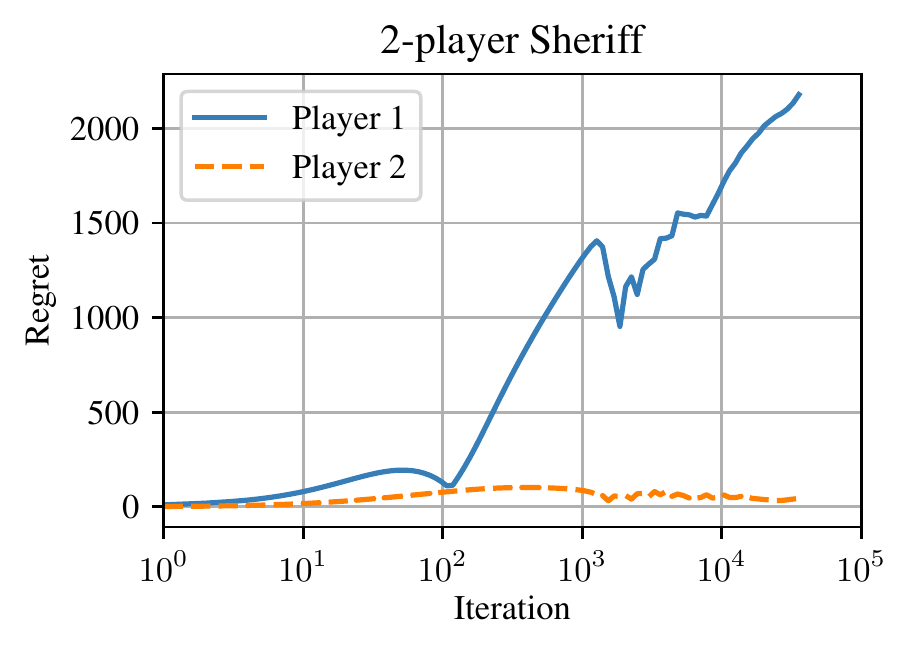}\\
    \caption{The regret of the players when they follow our learning dynamics, \algoshort; after a very mild tuning process, we selected the same learning rate $\eta \defeq 0.5$ for all games. The $x$-axis indexes the iteration, while the $y$-axis the regret. The scale on the $x$-axis is \emph{logarithmic}. We observe that the regret of each player grows as $O(\log T)$, verifying \Cref{theorem:main-detailed}.}
    \label{fig:kuhn}
\end{figure}

\end{document}